\documentclass{iopjournal} % {{{*
\usepackage[T1]{fontenc}
\usepackage{lmodern}
\usepackage{physics}
\usepackage{graphicx}
\usepackage{bm}
\usepackage{dsfont}
\usepackage{amssymb}
\usepackage{amsthm}
\usepackage{mathtools}
\allowdisplaybreaks
\usepackage{xcolor}
\definecolor{jacolor}{RGB}{200,40,0}
\definecolor{FPcolor}{HTML}{003CC8}
\definecolor{mgcolor}{RGB}{128,0,128}

\usepackage[final,inline,nomargin]{fixme} \fxsetup{theme=color}
\FXRegisterAuthor{ja}{aja}{\color{jacolor}JA}
\FXRegisterAuthor{af}{aaf}{\color{olive}AF}
\FXRegisterAuthor{fp}{afp}{\color{FPcolor}FP}
\FXRegisterAuthor{pca}{apc}{\color{teal}PCA}
\FXRegisterAuthor{mg}{amg}{\color{mgcolor}MG}

\usepackage{hyperref}
\hypersetup{
  colorlinks=true,
  linkcolor=blue,
  filecolor=blue,      
  citecolor=blue,
  urlcolor=blue,
  pdftitle={Advances in the quantum-state texture theory},
  pdfauthor=author={Jose Alfredo de Leon, Miguel Gonzalez, Alejandro Fonseca, Pedro C. Azado, Fernando Parisio},
}

\usepackage{leftidx}
\newcommand{\Eref}[1]{Eq.~(\ref{#1})} 
\newcommand{\Sref}[1]{Sec.~\ref{#1}}

\newcommand{\cS}{\mathcal{S}}

\newcommand{\cj}{Choi-Jamio\l{}kowski}

\newcommand{\unam}{Universidad Nacional Aut\'onoma de M\'exico, Ciudad de M\'exico 04510, Mexico}
\newcommand{\ifunam}{Instituto de F\'{\i}sica, \unam}
\newcommand{\icn}{Instituto de Ciencias Nucleares, \unam}

\newtheorem{proposition}{Proposition}
\begin{document}

\articletype{Paper}
% Title, authors, and affiliations {{{*
\title{Advances in the quantum-state texture theory}

\author{
  Jose Alfredo de Leon$^{1,\dagger}$\orcid{0000-0003-0045-9017}, 
  Miguel Gonzalez$^2$\orcid{0009-0004-0112-5988},
  Alejandro Fonseca$^3$\orcid{0000-0003-2633-9734},
  Pedro C. Azado$^4$\orcid{0000-0001-7319-9098} and
  Fernando Parisio$^{4,*}$\orcid{0000-0001-5818-8366}
}

\affil{$^1$ \ifunam}

\affil{$^2$ \icn}

\affil{$^3$ Departamento de Física, Universidad Nacional de Colombia, 111321, Bogotá, Colombia}

\affil{$^4$ Departamento de
  F\'{\i}sica, Centro de Ci\^encias Exatas e da Natureza, Universidade Federal de Pernambuco, Recife, Pernambuco 50670-901 Brazil}

\affil{${}^{*}$ Author to whom any correspondence should be addressed.}

\email{${}^\dagger$\href{mailto:deleongarrido.jose@gmail.com}{deleongarrido.jose@gmail.com}, ${}^*$\href{mailto:fernando.parisio@ufpe.br}{fernando.parisio@ufpe.br}}


% *}}}
\begin{abstract} % {{{*
The recently introduced concept of quantum-state texture (QST) has found applications ranging from the identification of unknown quantum gates to the study of quantum phase transitions and criticality, and has already been experimentally investigated.
Here, we advance its resource-theoretic formulation in several directions.
Our approach centers on the experimentally accessible grand sum---the sum of all matrix elements of a density operator in a given basis.
We clarify several points raised in the recent literature, including the relation between texture distillation and a simple grand-sum-based QST monotone.
We completely characterize free operations for a single qubit and study the resulting state-conversion order; classify relevant families of free operations in arbitrary dimensions; derive bounds relating the texture of composite systems to that of their subsystems; and introduce the basis-independent concept of set texture.
Together, these results strengthen and broaden the resource theory of QST, providing a foundation for its further development and applications.
\end{abstract} % *}}}
\section{Introduction} % {{{* 1
Coherent superpositions lie behind a vast range of quantum effects, with fundamental implications on both the foundations of quantum physics and emerging technologies. 
As several non-classical phenomena, in the last decade, coherence has been studied under the standpoint of quantum resource theories, a cornerstone of modern quantum information science. By partitioning quantum states and operations into ``free'' elements (those easily accessible under specific physical constraints) and ``resources'' (which cannot be generated under those constraints), quantum resource theories provide a rigorous mathematical framework for quantifying the utility of physical systems. This structural approach has been applied to various quantum properties, including quantum entanglement, quantum coherence, asymmetry, and non-stabilizerness~\cite{gour}. Each framework provides distinct insights into the foundational limitations and operational capabilities of quantum systems under specific physical constraints.

The characterization of different aspects of quantum coherence, such as its direct quantification as a formal resource \cite{l1}, its degree of ``imaginarity'' \cite{imag}, and genuine multilevel coherence \cite{ringbauer}, has proved useful across different fields. More recently, a related concept was developed with the introduction of the quantum-state texture (QST) resource theory \cite{qst}. This quantity has direct algebraic and geometric interpretations and is experimentally accessible~\cite{QSTexperiment}.

The first application of QST appeared in Ref.~\cite{qst} and consists of a protocol to identify unknown quantum gates without the need of full state tomography.
Since then, many other aspects have been explored. It has been noticed in Ref.~\cite{salazar} that, because the resourceless set in QST corresponds to a single pure state (being an extremal point), there are non-trivial constraints on candidate monotones in generic frameworks intended to encompass generic resource theories.
Many (but not all) features of QST resource theory do not rely on the specific form of the textureless state but rather on the fact that it is a pure state. This observation has been used by the authors of Ref.~\cite{generalization} to propose a unified ``fixed-point'' resource theory and use this generalization in a protocol to identify unknown gates in quantum circuit layers, following the general lines of the original scheme presented in Ref.~\cite{qst}.

Several works building on QST quantification have also appeared. These works can be split into two large classes: those introducing measures based on the grand sum (the sum of all density matrix entries) \cite{qst,tinggui}, and those proposing other ways to quantify QST \cite{tinggui,cao,kim,Muthu,yu,lei,mondal}. 

Interesting connections between QST and other resources, such as purity, non-stabilizerness, and entanglement, have been reported in Ref.~\cite{aditi}. In the same work, a surprising relation between QST and quantum phase transitions is described in the context of Ising models \cite{aditi}. The authors showed that the texture of ground states acts as a precise indicator for quantum phase transitions in the transverse Ising model, revealing critical points where several traditional markers often remain subtle. This line of research has recently been further investigated in Ref.~\cite{lucas} for dynamical criticality.

The concept of QST has also appeared in a variety of other contexts, for example, in relation to the capacity of quantum batteries \cite{bateries}; in relativistic quantum field theory \cite{relativity,lorentz,BH,zhang2026}; in addressing the principle of superposition \cite{sup}; and in general frameworks for resource quantification \cite{mondal}.

In the present work, we advance the theory of QST in several directions. In \Sref{sec:preliminaries}, we review the basic definitions and grand-sum-based measures of QST and discuss their operational meaning. In \Sref{sec:free-qubit-operations}, we characterize the complete set of free operations for a single qubit, and in \Sref{sec:partial-ordering}, we study the conditions under which one qubit state can be freely transformed into another. In \Sref{sec:survey-free-operations}, we consider higher-dimensional systems and present examples of operations that preserve, reduce, or erase texture. In \Sref{sec:bipartite-systems}, we establish relations between the texture of a bipartite system and that of its subsystems. Finally, in \Sref{sec:set-texture}, we extend QST from individual states to families of states and compare different ways of quantifying their texture in a common basis.

% *}}}
\section{Preliminaries} % {{{* 2
\label{sec:preliminaries}
For a fixed basis $\{|j\rangle\}$, referred to as the computational basis, there is a single quantum state for which all matrix elements equal the same constant value:
\begin{equation}
|f_1\rangle=\frac{1}{\sqrt{D}}\sum_{j=1}^D|j\rangle,
\end{equation}
where $D$ is the dimension of the corresponding Hilbert space, ${\cal H}$.
In the resource theory of QST, this textureless or flat state (see Fig.~\ref{fig1}) corresponds to the entire resourceless set.

Therefore, the single additional requirement for an arbitrary completely positive and trace preserving (CPTP) map $\Lambda$ to be a free operation is to have $|f_1\rangle$ as a fixed point, $\Lambda(f_1)=f_1$, with $f_1=|f_1\rangle \langle f_1|$.
Interestingly, the other states in the Fourier basis 
\begin{equation}
\label{fourier}
|f_k\rangle=\frac{1}{\sqrt{D}}\sum_{j=1}^D \omega_D^{(k-1)(j-1)}|j\rangle,
\end{equation}
maximize QST, as well as any linear combination of them, with $\omega_D=e^{2\pi i/D}$ and $k \ge 2$. 

In terms of QST quantification, given an arbitrary QST measure, ${\cal T}$, the following properties must hold:
\begin{itemize}
\item Vanishing at the flat state: ${\cal T}(f_1)=0$,
\item Monotonicity: ${\cal T}(\varrho)\ge {\cal T}(\Lambda(\varrho))$ for all $\varrho$ in the Hilbert-Schmidt space,
\item Convexity: state texture must not increase under mixing ${\cal T}(\sum p_i \varrho_i) \le \sum p_i {\cal T}(\varrho_i)$. 
\end{itemize}
%
%Note that the non-negativity of ${\cal T}$ follows from (i) and (ii). 
In Ref.~\cite{qst} and its supplemental material it is shown that any well-behaved decreasing function of
\begin{equation}
\label{GS}
D \langle f_1|\varrho|f_1 \rangle= \sum_{i, j=1}^D\varrho_{ij} \equiv \Sigma(\varrho).
\end{equation}
Hereafter we refer to $\Sigma(\varrho)$ as the grand sum of $\varrho$ in the chosen basis. It is immediately clear that %
\begin{equation}
\Sigma\left(\sum p_i \varrho_i\right) = \sum p_i \Sigma(\varrho_i),
\end{equation}
that is, the grand sum is invariant under mixing. Because of this, any convex function ${\cal T}$ of $\Sigma$ is convex. From the previous equation we have
\begin{equation}
{\cal T}\left[\Sigma\left(\sum p_i \varrho_i\right)\right] = {\cal T}\left[\sum p_i \Sigma(\varrho_i)\right],
\end{equation}
and from Jensen's inequality \cite{ineq} we directly obtain:
\begin{equation}
{\cal T}\left[\Sigma\left(\sum p_i \varrho_i\right)\right] \le \sum p_i {\cal T}\left[\Sigma(\varrho_i)\right],
\end{equation}
which corresponds to the convexity of ${\cal T}[\Sigma(\varrho)]$ as a QST quantifier.

Therefore, any well-behaved convex decreasing function of $\Sigma$ that vanishes at $f_1$ is a proper QST quantifier \cite{qst}. This is immediate because the grand sum has been shown to be non-decreasing under free operations (making grand-sum-based QST non-increasing under free operations) \cite{qst}. If, e.g., one takes ${\cal T}(x)=-\ln{x}$, the resulting measure is the rugosity \cite{qst}. Another specific example of measures ${\cal T}$ satisfying all previous requirements is given by the following family:
\begin{equation}
T_{\mu, \nu}(\varrho)=(1-\langle f_1|\varrho|f_1 \rangle^{\mu})^{\nu}
\end{equation}
The functions $T_{\mu, \nu}(\varrho)$ correspond to valid QST measures for $0\le \mu\le 1$ and $\nu\ge 1$, providing particular examples of what is described in Ref.~\cite{qst}.

Yet, the simplest instance, with $\mu=\nu=1$, was addressed as a separate quantity in Ref.~\cite{tinggui} and named the geometric measure. Now we show that this quantity corresponds to the asymptotic distillation rate of QST, $r$, in a fairly simple protocol. Suppose that one has a large number $N$ of copies of an arbitrary quantum state $\sigma$. From $\varrho=\sigma^{\otimes N}$ one intends to extract $M\le N$ copies of maximum-texture states, say $f_2\equiv |f_2\rangle \langle f_2|$. All we have to do is make projective measurements in the Fourier basis. Every time the result $f_1$ is obtained (with probability $\langle f_1|\sigma|f_1 \rangle$), the system is discarded. For all other outcomes, a maximum-texture state is obtained and, for each of them, there is a free unitary taking $f_k$ to $f_2$, $k\ge 2$. Therefore, the asymptotic number of distilled states is $M=(1-\langle f_1|\sigma|f_1 \rangle)N=T(\sigma)N$, with $T(\sigma)\equiv T_{1,1}(\sigma)$. Therefore,
\begin{equation}
r=\frac{M}{N}\longrightarrow T(\sigma), \;\;N \rightarrow \infty.
\end{equation}
This provides an operational interpretation of the monotone $T$ as the ``distillable texture'' (or ``distillable QST'').

Another important QST monotone is the rugosity
\begin{equation}
\label{rug}
\mathfrak{R}(\varrho)=-\ln \left(\frac{\Sigma(\varrho)}{D}\right)=-\ln\langle f_1|\varrho|f_1 \rangle,
\end{equation}
corresponding to a quantifier that is additive under tensoring. That is $\mathfrak{R}(\varrho^{\otimes n})=n\mathfrak{R}(\varrho)$ \cite{qst}. This is an important property, especially when dealing with systems composed of a large number of parties, as illustrated in its application to spin chains \cite{aditi,lucas}.
We remark that, in general, additivity ceases to be valid if one defines less symmetric states as resourceless, as is done in Ref.~\cite{generalization}. This is because $f_1$ is preserved by tensor products: $|f_1^{(D)}\rangle=|f_1^{(D_1)}\rangle\otimes |f_1^{(D_2)}\rangle$, where $D=D_1\times D_2$, $D_1$ and $D_2$ are the dimensions of the Hilbert spaces of subsystems 1 and 2, respectively.

\begin{figure}
\begin{center}
\includegraphics[height=5cm]{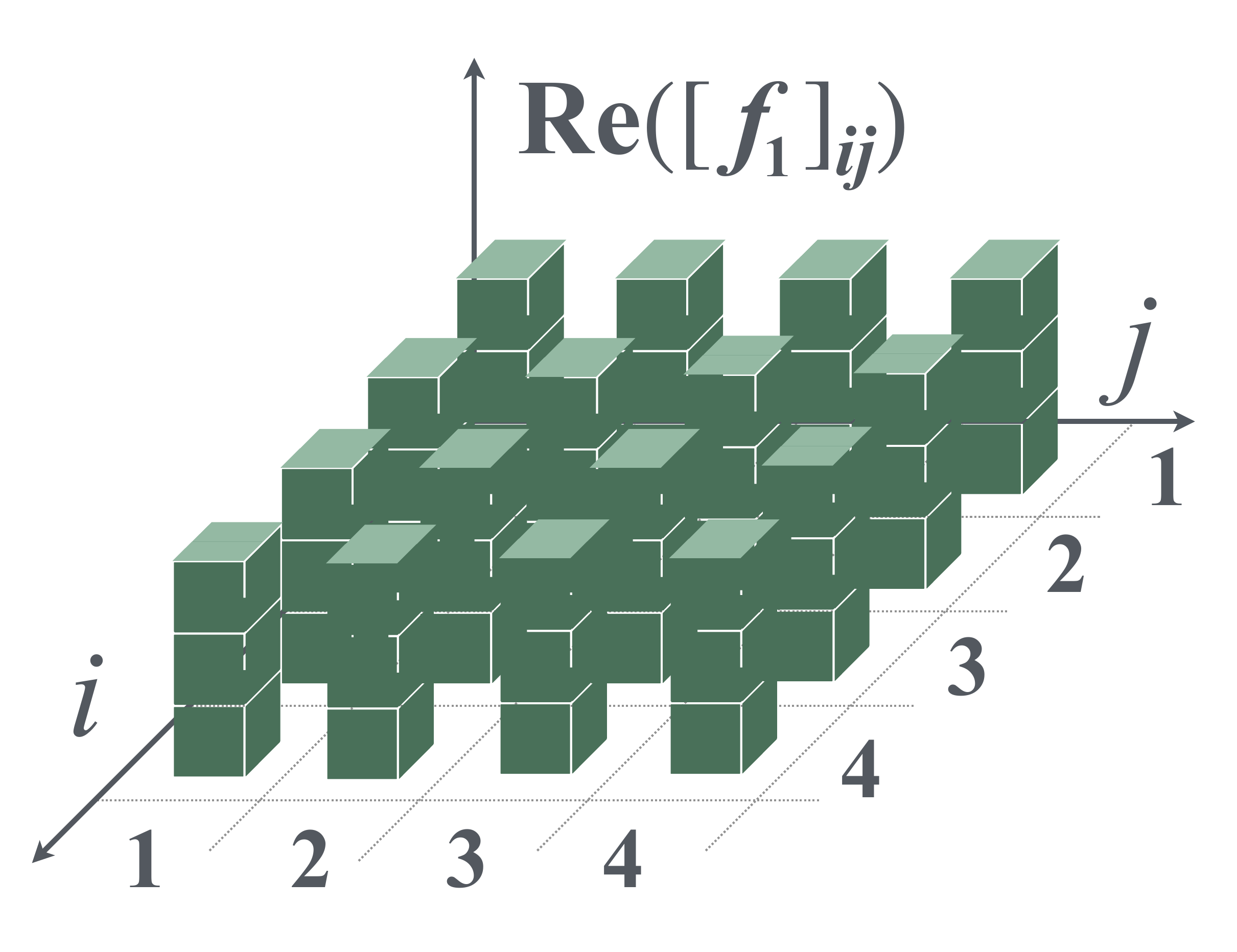}
\includegraphics[height=5cm]{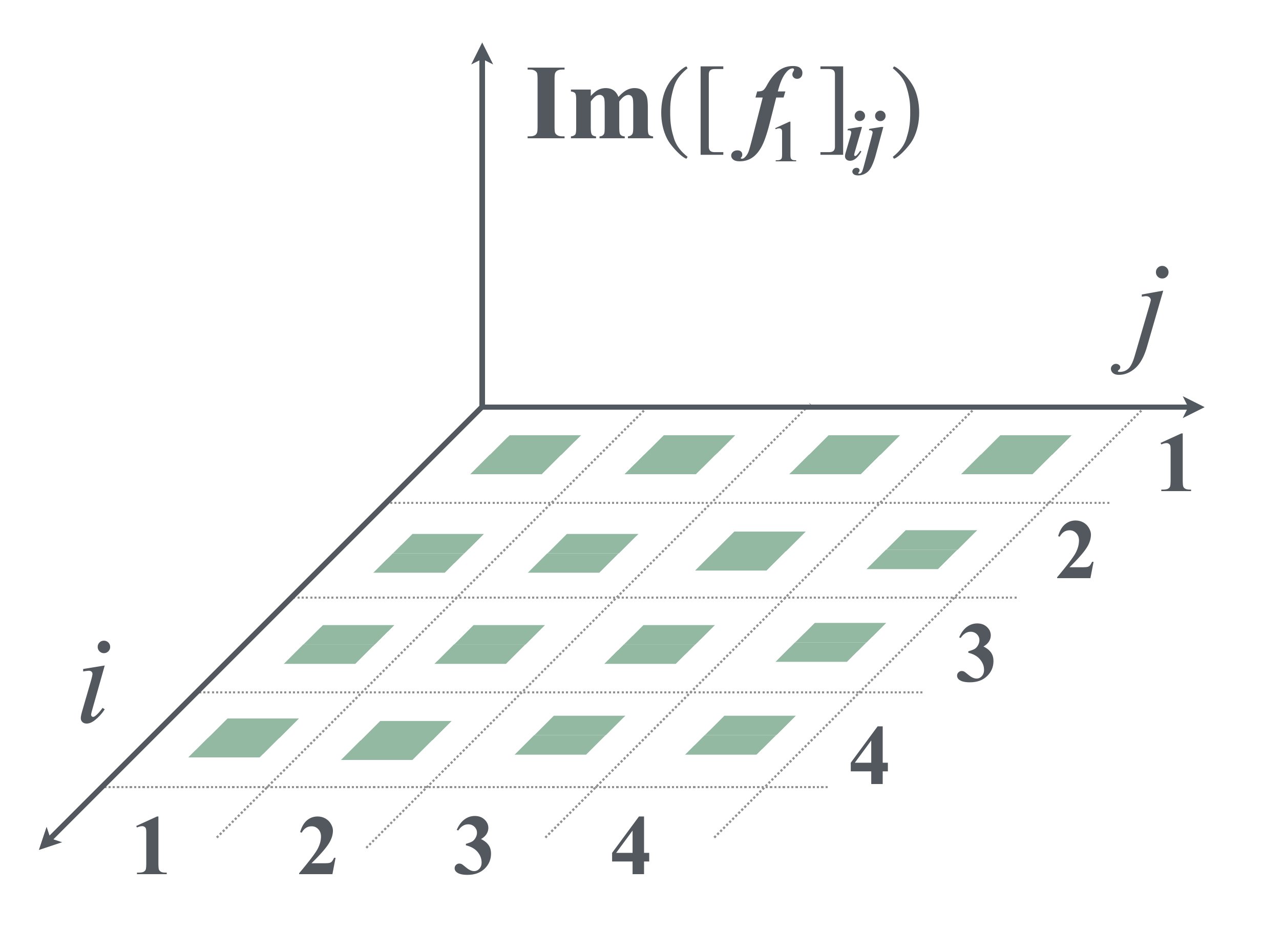}
  \caption{Pictorial 3D view of the real (left) and imaginary (right) parts of $f_1$ matrix elements. This is the only quantum state for which both plots are flat. }
  \label{fig1}
    \end{center}
\end{figure}
Suppose, for instance, that in the context of ``fixed point'' resource theories \cite{generalization}, one chooses the resourceless state as the one for which the probability of finding the system in the ground state is twice the probability of observing it in the first excited state which, in turn, is twice as large as the probability to find it in the second excited state, and so on. If one assumes that the coefficients are positive real numbers, this state, for a qubit, reads $|g_1^{(2)}\rangle=(\sqrt{2}|0\rangle+|1\rangle)/\sqrt{3}$, and therefore $|g_1^{(2)}\rangle\otimes |g_1^{(2)}\rangle=(2|0\rangle+\sqrt{2}|1\rangle+\sqrt{2}|2\rangle+|3\rangle)/3$, where we set $|00\rangle=|0\rangle$, $|10\rangle=|1\rangle$, $|01\rangle=|2\rangle$, and $|11\rangle=|3\rangle$. However, we should have $|g_1^{(4)}\rangle=(2\sqrt{2}|0\rangle+2|1\rangle+\sqrt{2}|2\rangle+|3\rangle)/\sqrt{15}$. This illustrates the fact that, in general, $|g_1^{(D)}\rangle\ne|g_1^{(D_1)}\rangle\otimes |g_1^{(D_2)}\rangle$. 
Therefore, for general fixed-point resource theories the additivity of rugosity is lost.

It should also be noted that the grand sum can be expressed as $\Sigma(\varrho)= \sum_{i, j=1}^D{\rm Re}(\varrho_{ij})$, since $\varrho$ is Hermitian. Therefore, the statement that texture is ``sensitive to imaginarity'' \cite{qst} needs clarification. The point is that, for a fixed amount of coherence, $\sum_{i, j=1}^D|\varrho_{ij}|=const$., there is a continuum of possible texture values, depending on the imaginarity of the particular state. 

%============== Grand sum and stabilizerness ==========
%=======================================================
Here we note a direct stabilizer interpretation of the grand sum itself in addition to the connections already established between non-stabilizerness and QST~\cite{aditi}. For an \(N\)-qubit system of dimension \(D=2^N\), the textureless state in the computational basis is
\(|f_1\rangle=|+\rangle^{\otimes N}\), where
\(|+\rangle=(|0\rangle+|1\rangle)/\sqrt{2}\) and $\sigma_x\ket{+}=\ket{+}$. $\ket{f_1}$ is therefore a stabilizer state in the standard stabilizer formalism \cite{Gottesman1997}: specifically, it is the unique joint \(+1\) eigenstate of the \(N\) independent commuting generators

\begin{equation}
    \sigma^i_x|f_1\rangle=|f_1\rangle,
    \qquad i=1,\ldots,N.
\end{equation}

These generators define its stabilizer group

\begin{equation}
    \mathcal{S}_{f_1}
    =
    \langle \sigma_x^1,\ldots,\sigma_x^N\rangle
    =
    \{I,\sigma_x\}^{\otimes N}.
\end{equation}

Its projector therefore has the standard stabilizer representation

\begin{equation}
    P_{f_1}
    =
    |f_1\rangle\langle f_1|
    =
    \frac{1}{D}
    \sum_{S\in\mathcal{S}_{f_1}} S.
\end{equation}

Substituting in the grand sum~\ref{GS} the projector identity gives

\begin{equation}
    \Sigma(\varrho)
    =
    \sum_{S\in\mathcal{S}_{f_1}}
    \langle S\rangle_{\varrho}.
\end{equation}

Thus, the grand sum is exactly the expectation value of the stabilizer-group sum associated with the textureless state. A similar expression in the context of phase-transitions in quantum spin chains is given in Ref.~\cite{spinsphasetransition2026}.

%Equivalently, \(\Sigma(\rho)/D\) is the probability of obtaining the all-\(+1\) (trivial) syndrome in a simultaneous measurement of the commuting observables \(X_1,\ldots,X_N\).

%===========================================================

A final point worth noting is the following. From Eq. (\ref{GS}), it is clear that, for any physical state $\varrho$, we must have
\begin{equation}
\label{positive1}
\Sigma(\varrho)\ge0. 
\end{equation}
This is a necessary (but not sufficient) condition for positivity, which is easily computable. For instance, consider the hypothetical density operators given by 
\begin{equation}
\label{state}
\varrho=\frac{1}{D}\sum |i\rangle \langle i|-a\sum_{i\ne j}|i\rangle \langle j|
\end{equation}
with $a$ real and positive. The previous condition imposes $1-(D^2-D)a\ge 0$, that is $a\le 1/(D^2-D)$. Another well-known necessary condition \cite{cohen} reads
\begin{equation}
\label{positive2}
\varrho_{ii}\varrho_{jj}\ge\abs{\varrho_{ij}}^2.
\end{equation}
We note that (\ref{positive1}) and (\ref{positive2}) are able to detect non-positivity in different situations. For instance, the obviously non-positive operator ${\rm diag}[-1/2,+1/2,+1/2,+1/2]$, would pass condition (\ref{positive1}) but fail (\ref{positive2}).
On the other hand, the state (\ref{state}) would pass condition (\ref{positive2}) for $a\le 1/D$. Therefore, the two conditions are not equivalent.

% In the next section, we address the partial ordering induced by any quantifier based on the grand sum.

% *}}}
\section{Characterization of single-qubit free operations} % {{{* 3
\label{sec:free-qubit-operations}

We now characterize the set of QST free operations for a single qubit. An arbitrary qubit state can be written as
\begin{equation}\label{eq:qubit-bloch}
  \varrho =
  \frac{1}{2} \left(\mathds{1}+ x\sigma_x + y\sigma_y + z\sigma_z\right),
\end{equation}
where $\sigma_x$, $\sigma_y$, and $\sigma_z$ are the Pauli matrices, and $\bm r=(x,y,z)$ is the Bloch vector, satisfying $\norm{\bm r}\leq1$. Any single-qubit trace-preserving and positive, but not necessarily completely positive, map acts affinely on the Bloch ball as $\bm r\mapsto M\bm r+\bm t$, where $M$ is a real $3\times3$ matrix and $\bm t$ is a real vector~\cite{bethruskai_2002_analysis}. Thus, a QST free operation must satisfy two additional conditions: (1) it must preserve the textureless state, $\Lambda(f_1)=f_1$, and (2) it must be completely positive. We first derive the constraints imposed by these conditions, then give a geometric characterization of the resulting operations, and finally determine how they transform the grand sum.

The unnormalized \cj{} matrix of a free operation $\Lambda$ is an operator on the doubled Hilbert space $\mathcal H\otimes\mathcal H$. Using the basis $\{\ket{f_1},\ket{f_2}\}$, it reads
\begin{equation}\label{eq:cj_free_ops}
  J_\Lambda = \sum_{j,k=1}^{2} \dyad{f_j}{f_k}\otimes\Lambda(\dyad{f_j}{f_k}).
\end{equation}
Hermiticity of $J_\Lambda$, trace preservation, and the fixed-point condition $\Lambda(f_1)=f_1$ determine its block structure. The details of this derivation are given in Appendix~\ref{app:free-qubit-choi}. Therefore, before imposing complete positivity, the most general \cj{} matrix compatible with these constraints can be written in the ordered basis $\{\ket{f_1\otimes f_1},\ket{f_1\otimes f_2},
\ket{f_2\otimes f_1},\ket{f_2\otimes f_2}\}$ as
\begin{equation}
  J_\Lambda =
  \mqty(
    1 & 0 & a & b\\
    0 & 0 & c & -a \\
    a^* & c^* & 1-\eta & q \\
    b^* & -a^* & q^* & \eta
  ),
\end{equation}
where $\eta$ is real and $a,b,c,q$ are, in general, complex.

Complete positivity is equivalent to the \cj{} matrix being positive semidefinite, $J_\Lambda\geq0$~\cite{bengtsson_2006_geometry}. For a Hermitian matrix, this is equivalent to the non-negativity of all its principal minors, i.e., the determinants of the submatrices obtained by selecting the same set of rows and columns. Since the second diagonal entry of $J_\Lambda$ vanishes, the $2\times2$ principal minor associated with indices $2$ and $k$ is
\begin{equation}
  \det
  \mqty(
    0 & (J_\Lambda)_{2k} \\
    (J_\Lambda)_{2k}^* & (J_\Lambda)_{kk}
  )
  =
  -\abs{(J_\Lambda)_{2k}}^2.
\end{equation}
Non-negativity of this minor therefore requires $(J_\Lambda)_{2k}=0$ for every $k$. Applied to the matrix above, this gives $a=c=0$. Renaming $b=\lambda$, every free single-qubit operation has a \cj{} matrix of the form
\begin{equation}
  J_\Lambda =
  \mqty(
    1 & 0 & 0 & \lambda \\
    0 & 0 & 0 & 0 \\
    0 & 0 & 1-\eta & q \\
    \lambda^* & 0 & q^* & \eta
  ).
  \label{eq:free-qubit-choi}
\end{equation}
This already implies that one eigenvalue of $J_\Lambda$ is zero. Thus, the maximum Kraus rank of a QST free operation is three.

It remains to determine which values of $\eta$, $\lambda$, and $q$ make Eq.~\eqref{eq:free-qubit-choi} positive semidefinite. Since its second row and column vanish, it is sufficient to require the remaining $3\times3$ principal block to be positive semidefinite:
\begin{equation}
  \mqty(
    1 & 0 & \lambda \\
    0 & 1-\eta & q \\
    \lambda^* & q^* & \eta
  ).
\end{equation}
Requiring all its principal minors to be non-negative yields the complete set of constraints:
\begin{equation}\label{eq:free-qubit-cp}
  0\leq\eta\leq1,
  \qquad
  \abs{\lambda}^2\leq\eta,
  \qquad
  \abs{q}^2 \leq (1-\eta) \qty(\eta-\abs{\lambda}^2).
\end{equation}

Importantly, this characterization does not depend on the explicit representation of $\ket{f_1}$ in the computational basis. The derivation only uses the basis $\{\ket{f_1},\ket{f_2}\}$ adapted to the fixed state. We specialize to the QST computational basis only below, when describing the action of free operations geometrically on the Bloch ball.

To obtain a geometric characterization of QST free operations, we now fix the computational basis, so that the corresponding Fourier basis is $\{\ket{f_1},\ket{f_2}\}=\{\ket{+},\ket{-}\}$. In this basis, a single-qubit state with Bloch vector $\bm r=(x,y,z)$ reads
\begin{equation}
  \varrho =
  \frac{1}{2}
  \qty[
    (1+x)\dyad{+}{+}
    +(z+iy)\dyad{+}{-}
    +(z-iy)\dyad{-}{+}
    +(1-x)\dyad{-}{-}
  ].
\end{equation}

Let $\bm r'=(x',y',z')$ denote the Bloch vector of the transformed state. Using the block structure of the \cj{} matrix in Eq.~\eqref{eq:free-qubit-choi} and the linearity of $\Lambda$, the components of $\bm r'$ are
\begin{equation}
  \begin{aligned}
    x' &= 1-\eta+\eta x, \\
    y' &= \lambda_R y+\lambda_I z+q_I(1-x), \\
    z' &= -\lambda_I y+\lambda_R z+q_R(1-x),
  \end{aligned}
\end{equation}
where $\lambda=\lambda_R+i\lambda_I$ and $q=q_R+iq_I$. Thus, every single-qubit QST free operation acts on the Bloch vector through the affine map
\begin{equation} \label{eq:free-qubit-affine-map}
  \bm r' =
  \mqty(
    \eta & 0 & 0 \\
    -q_I & \lambda_R & \lambda_I \\
    -q_R & -\lambda_I & \lambda_R
  )
  \bm r +
  \mqty(
    1-\eta \\
    q_I \\
    q_R
  ).
\end{equation}

\begin{figure}
  \centering
  \includegraphics[width=0.4\textwidth]{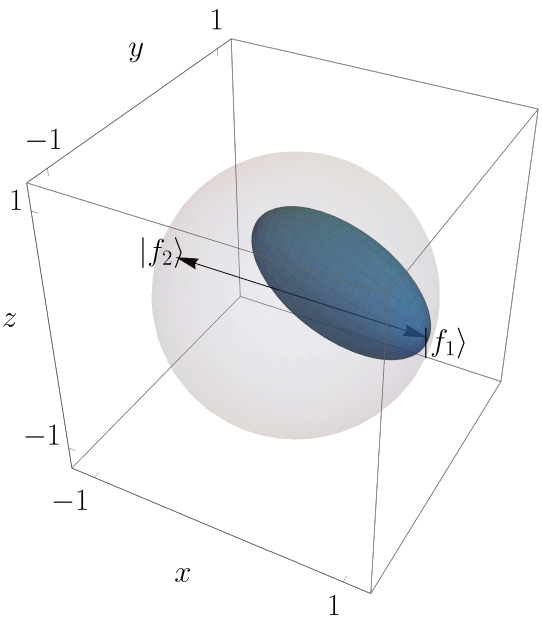}
  \caption{Geometric action of a single-qubit QST free operation; see \Eref{eq:free-qubit-affine-map}. The boundary of the Bloch ball is shown as the outer sphere, while its image under the free operation is the inner ellipsoid. The states $\ket{f_1}=\ket{+}$ and $\ket{f_2}=\ket{-}$ correspond to the Bloch vectors $(1,0,0)^T$ and $(-1,0,0)^T$, respectively. The parameters are $\eta=4/5$, $\lambda=(2/5)e^{7i/10}$, and $q=(1/4)e^{7i/5}$.}
  \label{fig:single-qubit_free-ops}
\end{figure}

To make the geometric action of a free operation explicit, we write $\lambda=\abs{\lambda} e^{-i\theta}$ and define
\begin{equation}
  \bm r_\perp = \mqty(y \\ z),
  \qquad
  \bm q = \mqty(q_I \\ q_R),
  \qquad
  R_\theta =
  \mqty(
    \cos\theta & -\sin\theta \\
    \sin\theta & \cos\theta
  ).
\end{equation}
With these definitions, the action of every single-qubit free operation separates into the transformation of the longitudinal coordinate $x$ and that of the transverse vector $\bm r_\perp$:
\begin{equation}\label{eq:free-qubit-bloch-action}
  x' =1-\eta+\eta x, 
  \qquad
  \bm r_\perp' = 
  \abs{\lambda} R_\theta\bm r_\perp +
  (1-x)\bm q,
\end{equation}
while the complete-positivity conditions in Eq.~\eqref{eq:free-qubit-cp} become
\begin{equation}
  0\leq\eta\leq1,
  \qquad
  0\leq\abs{\lambda}\leq\sqrt{\eta},
  \qquad
  \norm{\bm q}^2 \leq
  (1-\eta)
  \qty(
    \eta-\abs{\lambda}^2
  ).
  \label{eq:free-qubit-geometric-constraints}
\end{equation}

Equations~\eqref{eq:free-qubit-bloch-action} and \eqref{eq:free-qubit-geometric-constraints} provide a geometric characterization of the complete set of free operations. A plane of constant $x$ is mapped into another plane of constant $x'$. Within each such plane, the transverse Bloch vector $\bm r_\perp$ is rotated by $\theta$, contracted isotropically by a factor $\abs{\lambda}$, and translated by $(1-x)\bm q$. The translation vanishes at $x=1$, ensuring that $(1,0,0)^T$ remains fixed. As illustrated in Fig.~\ref{fig:single-qubit_free-ops}, the image of the Bloch ball is therefore an ellipsoid, possibly degenerate, contained in the Bloch ball and containing $f_1$ on its boundary.

The characterization above relies only on the fixed-point condition
$\Lambda(f_1)=f_1$ and complete positivity, and is therefore independent of
the grand sum. We now determine how the grand sum transforms under these
free operations. From \Eref{GS}, for a single qubit,
\begin{equation}
  \Sigma(\varrho)=2\ev{\varrho}{+}=1+x.
  \label{eq:qubit-GS}
\end{equation}
Then, \Eref{eq:free-qubit-bloch-action} gives
\begin{equation}\label{eq:free-qubit-grand-sum}
  \Sigma\qty[\Lambda(\varrho)] = 2(1-\eta) + \eta\Sigma(\varrho).
\end{equation}
Therefore,
\begin{equation}
  \Sigma\qty[\Lambda(\varrho)] - \Sigma(\varrho) =
  (1-\eta)\qty[ 2-\Sigma(\varrho) ]\geq 0,
\end{equation}
showing that the grand sum is non-decreasing under free operations. 

% *}}}
\section{Partial ordering and state convertibility} % {{{* 4
\label{sec:partial-ordering} 
Having characterized the single-qubit QST free operations, we now study the
state transformations they allow and the partial ordering induced by the
grand sum. Different QST monotones may, in general, induce different
orderings; here we focus specifically on the ordering associated with
grand-sum-based quantification. As shown in the previous section, the grand
sum is non-decreasing under free operations. Therefore, a necessary condition
for a free operation $\Lambda$ to transform $\varrho$ into $\tau$ is
\begin{equation}
  \Sigma\qty(\tau)\geq \Sigma(\varrho).
  \label{eq:GS-free-monotonicity}
\end{equation}
We now ask whether this condition is also sufficient. For pure states, we show below that it is; for general mixed states, however, it is not. We establish these two results in this order in the remainder of the section.

\begin{proposition}
\label{prop:pure-state-conversion}
Let
$\varrho=|\psi\rangle\langle\psi|$
and
$\tau=|\phi\rangle\langle\phi|$
be pure single-qubit states. There exists a free operation $\Lambda$ such that $\Lambda(\varrho)=\tau$ if and only if
\begin{equation}
\Sigma(\tau)\geq \Sigma(\varrho).
\label{eq:pure-GS-order}
\end{equation}
\end{proposition}

\begin{proof} % {{{*
The necessity of Eq.~\eqref{eq:pure-GS-order} follows from the fact that the grand sum is non-decreasing under free operations. We now prove sufficiency. The goal is to construct a free operation explicitly that maps $\varrho$ to $\tau$. We do this by introducing an auxiliary environment, defining an isometry on the system and environment, and then tracing out the environment.

Since $D=2$, the grand sum of the two pure states is
$ \Sigma(\varrho) = 2\abs{\braket{f_1}{\psi}}^2$ and $\Sigma(\tau) = 2\abs{\braket{f_1}{\phi}}^2$, respectively. Therefore, Eq.~\eqref{eq:pure-GS-order} is equivalent to
\begin{equation}
\left|\langle f_1|\phi\rangle\right|
\geq
\left|\langle f_1|\psi\rangle\right|.
\label{eq:pure-overlap-order}
\end{equation}
This inequality says that the final state has at least as much overlap with the free state $f_1$ as the initial state. This is precisely the condition that will allow the following construction.

First, consider two simple cases. If $|\psi\rangle=|f_1\rangle$, then Eq.~\eqref{eq:pure-overlap-order} implies that $|\phi\rangle=|f_1\rangle$ up to a global phase, so the identity channel realizes the transformation. If $\langle f_1|\phi\rangle=0$, then Eq.~\eqref{eq:pure-overlap-order} implies $\langle f_1|\psi\rangle=0$. Since the subspace orthogonal to $|f_1\rangle$ is one-dimensional for $D=2$, the two density operators are equal, $\varrho=\tau$, and the identity channel again realizes the transformation.

Second, we treat the nontrivial case in which $\ket{\psi}\neq\ket{f_1}$ and $\braket{f_1}{\phi}\neq0$. We introduce an auxiliary environment $E$, and define a map $V$ on the two input vectors by
\begin{equation}
\begin{aligned}
V\ket{f_1}
&=
\ket{f_1}\ket{e_0},
\\
V\ket{\psi}
&=
\ket{\phi}\ket{e_1},
\end{aligned}
\label{eq:pure-isometry}
\end{equation}
where $\ket{e_0}$ and $\ket{e_1}$ are normalized states of the environment, to be fixed below.

For $V$ to be an isometry, it must preserve the inner product between the two input vectors. Therefore, we need
\begin{equation}
\mel{f_1}{V^\dagger V}{\psi}
=
\braket{f_1}{\psi}.
\label{eq:isometry-condition}
\end{equation}
Using Eq.~\eqref{eq:pure-isometry}, the left-hand side is
\begin{equation}
\begin{aligned}
\mel{f_1}{V^\dagger V}{\psi}
&=
\left(
\bra{f_1}\bra{e_0}
\right)
\left(
\ket{\phi}\ket{e_1}
\right)
\\
&=
\braket{f_1}{\phi}
\braket{e_0}{e_1}.
\end{aligned}
\label{eq:isometry-inner-product}
\end{equation}
Thus, Eq.~\eqref{eq:isometry-condition} is satisfied if
\begin{equation}
\braket{e_0}{e_1}
=
\frac{\braket{f_1}{\psi}}
{\braket{f_1}{\phi}}.
\label{eq:environment-overlap}
\end{equation}
This choice is possible because Eq.~\eqref{eq:pure-overlap-order} guarantees 
that $\abs{\braket{f_1}{\psi} / \braket{f_1}{\phi}} \leq 1$.
Hence, one can choose normalized environment states $\ket{e_0}$ and $\ket{e_1}$ with the overlap required in Eq.~\eqref{eq:environment-overlap}. With this choice, $V$ preserves the Gram matrix of the two input vectors and can be extended to an isometry from the qubit Hilbert space into the joint system--environment Hilbert space.

The quantum channel we are looking for is then obtained by applying this isometry and discarding the environment:
\begin{equation}
\Lambda(\omega)
=
\operatorname{Tr}_E
\left(
V\omega V^\dagger
\right).
\label{eq:pure-free-channel}
\end{equation}
Since it is obtained from an isometry followed by a partial trace, $\Lambda$ is CPTP. Moreover, Eq.~\eqref{eq:pure-isometry} implies
\begin{equation}
\Lambda(f_1)=f_1,
\qquad
\Lambda(\varrho)=\tau.
\label{eq:pure-free-action}
\end{equation}
Thus, $\Lambda$ is a free operation that realizes the desired transformation from $\varrho$ to $\tau$. Therefore, Eq.~\eqref{eq:pure-GS-order} is sufficient. 
\end{proof} % *}}}

Proposition~\ref{prop:pure-state-conversion} has a simple geometric interpretation. If two single-qubit pure states $\ket{\psi}$ and $\ket{\phi}$ have the same grand sum, $\Sigma(\dyad{\psi})=\Sigma(\dyad{\phi})$, Eq.~\eqref{eq:qubit-GS} places them on the same circle of constant $x$ on the Bloch sphere. A free unitary must leave $|+\rangle$ invariant, and therefore, it has the form
\begin{equation}
U_x(\theta)
=
\exp\left(
-\frac{i\theta}{2}\sigma_x
\right).
\label{eq:free-x-rotation}
\end{equation}
These unitaries are rotations around the $x$ axis, so a suitable $U_x(\theta)$ connects any two pure states with the same grand sum. If instead
$\Sigma(\dyad{\psi})<\Sigma(\dyad{\phi})$, no free unitary can realize the transformation because Eq.~\eqref{eq:free-x-rotation} preserves $x$ and therefore preserves the grand sum. Proposition~\ref{prop:pure-state-conversion} shows that a non-unitary free operation nevertheless exists.

A particularly simple limiting case is the state $\ket{-}$, which is orthogonal to $\ket{f_1}=\ket{+}$ and has the minimum grand sum, $\Sigma=0$. In fact, $\ket{-}$ can be freely transformed into any single-qubit state $\tau$, pure or mixed. This follows from the measure-and-prepare channel
\begin{equation}
  \Lambda(\varrho) =
  \bra{+}\varrho\ket{+}\,f_1 +
  \bra{-}\varrho\ket{-}\,\tau, 
  \label{eq:minus-universal-channel}
\end{equation}
which is CPTP, satisfies $\Lambda(f_1)=f_1$, and therefore is free, while $\Lambda(\dyad{-})=\tau$. Thus, the state of maximum texture can reach any point of the Bloch ball through a free operation.

This special case, however, does not extend to arbitrary input states. Once general single-qubit states are considered, the grand sum remains a necessary condition for free conversion but is no longer sufficient to determine whether a transformation is possible.

\begin{proposition}
\label{prop:GS-not-sufficient}
There exist single-qubit states $\varrho$ and $\tau$ such that
\begin{equation}
\Sigma(\tau)>\Sigma(\varrho),
\label{eq:counterexample-GS-order}
\end{equation}
but no free operation $\Lambda$ satisfies $\Lambda(\varrho)=\tau$.
\end{proposition}

\begin{proof} % {{{*
We prove the statement by constructing a pair of states for which the grand sum increases, but no free operation can transform the first state into the second. The pair is chosen so that the grand sum increases, and therefore the monotonicity of the grand sum does not rule out the transformation. We then show that the transformation is nevertheless impossible, because it would increase the trace distance from the fixed state $f_1$, which no free operation can do. Consider the maximally mixed state
\begin{equation}
\varrho_0
=
\frac{\mathds{1}}{2}
\label{eq:mixed-counterexample-input}
\end{equation}
and the family of pure states
\begin{equation}
\tau_a
=
\frac{1}{2}
\left(
\mathds{1}
+a\sigma_x
+\sqrt{1-a^2}\,\sigma_y
\right),
\qquad
0<a<\frac{1}{2}.
\label{eq:mixed-counterexample-output}
\end{equation}
Using Eq.~\eqref{eq:qubit-GS}, their grand sums are
\begin{equation}
\Sigma(\varrho_0)=1,
\qquad
\Sigma(\tau_a)=1+a.
\label{eq:mixed-counterexample-GS}
\end{equation}
Hence, $\Sigma(\varrho_0)<\Sigma(\tau_a)$, so the monotonicity of the grand sum does not rule out the transformation $\varrho_0\mapsto\tau_a$.

Assume, for contradiction, that a free operation $\Lambda$ maps $\varrho_0$ to $\tau_a$. We use the trace distance, defined for two states $\omega$ and $\chi$ as
\begin{equation}
d_{\mathrm{tr}}(\omega,\chi)
=
\frac{1}{2}
\norm{\omega-\chi}_1,
\label{eq:trace-distance-definition}
\end{equation}
which is contractive under CPTP maps~\cite{watrous_2018_theory}:
\begin{equation}
d_{\mathrm{tr}}
\bigl(
\Lambda(\omega),
\Lambda(\chi)
\bigr)
\leq
d_{\mathrm{tr}}(\omega,\chi).
\label{eq:trace-distance-contractivity}
\end{equation}
Since $\Lambda$ is assumed to be free, $\Lambda(f_1)=f_1$. Together with the assumed transformation $\Lambda(\varrho_0)=\tau_a$, Eq.~\eqref{eq:trace-distance-contractivity} gives
\begin{equation}
\begin{aligned}
d_{\mathrm{tr}}(\tau_a,f_1)
% &=
% d_{\mathrm{tr}}
% \bigl(
% \Lambda(\varrho_0),
% \Lambda(f_1)
% \bigr)
% \\
&\leq
d_{\mathrm{tr}}(\varrho_0,f_1).
\end{aligned}
\label{eq:counterexample-contractivity}
\end{equation}

For qubit states $\omega$ and $\chi$ with Bloch vectors $\bm r_\omega$ and $\bm r_\chi$, the trace distance can be written as
\begin{equation}
d_{\mathrm{tr}}(\omega,\chi)
=
\frac{1}{2}
\norm{
\bm r_\omega-\bm r_\chi
}_2.
\label{eq:trace-distance-bloch}
\end{equation}
Thus, the trace distance between two qubit states is simply the Euclidean distance between their Bloch vectors multiplied by a factor of $\frac{1}{2}$. The Bloch vectors of the three states involved are
\begin{equation}
\bm r_{\varrho_0}=(0,0,0),
\qquad
\bm r_{f_1}=(1,0,0),
\qquad
\bm r_{\tau_a}
=
\left(
a,
\sqrt{1-a^2},
0
\right).
\label{eq:counterexample-bloch-vectors}
\end{equation}
Equation~\eqref{eq:trace-distance-bloch} then gives
\begin{equation}
d_{\mathrm{tr}}(\varrho_0,f_1)
=
\frac{1}{2},
\label{eq:counterexample-distance-input}
\end{equation}
whereas
\begin{equation}
d_{\mathrm{tr}}(\tau_a,f_1)
=
\sqrt{\frac{1-a}{2}}
>
\frac{1}{2},
\qquad
0<a<\frac{1}{2}.
\label{eq:counterexample-distance-output}
\end{equation}
Equations~\eqref{eq:counterexample-contractivity}, \eqref{eq:counterexample-distance-input}, and \eqref{eq:counterexample-distance-output} contradict each other. Hence, no free operation can map $\varrho_0$ to $\tau_a$, even though Eq.~\eqref{eq:mixed-counterexample-GS} gives $\Sigma(\varrho_0)<\Sigma(\tau_a)$.
\end{proof} % *}}}

The same pair makes the resulting partial ordering explicit. Proposition~\ref{prop:GS-not-sufficient} excludes the transformation $\varrho_0\mapsto\tau_a$, while Eq.~\eqref{eq:GS-free-monotonicity} excludes the reverse transformation $\tau_a\mapsto\varrho_0$, since it would decrease the grand sum. Thus, $\varrho_0$ and $\tau_a$ are incomparable under free operations of QST.

For pure single-qubit states, free convertibility is therefore completely determined by the grand sum, as established in Proposition~\ref{prop:pure-state-conversion}, but this characterization does not extend to the full Bloch ball. For arbitrary states, the grand sum gives a necessary but not sufficient condition for free conversion. Consequently, free convertibility is not a total relation on the state space; after identifying states that can be freely converted into one another, it gives the partial ordering considered here.

After completing this work, we became aware of Ref.~\cite{lin_2026_quantumstate}, whose Sec.~III.B establishes necessary and sufficient conditions for deterministic conversion between arbitrary qubit states under QST free operations. Both the pure-state conversion criterion and the insufficiency of grand-sum monotonicity for general states established in this section follow from their theorem. Our discussion provides alternative proofs and an explicit example of incomparable states, emphasizing the resulting partial ordering under
free convertibility.
% *}}}
\section{A survey on free operations} % {{{* 5
\label{sec:survey-free-operations}

In this section, we develop a compendium of sets of quantum maps that preserve, diminish, or completely deplete the amount of texture in quantum states, and provide several explicit examples involving quantum channels on systems of arbitrary dimension $D$. For simplicity of notation, throughout this section, all indices associated with the standard computational basis are taken to range from $0$ to $D-1$.

It is a well-known fact that a complete description of arbitrary density matrices associated with high-dimensional quantum systems may be carried out by employing the set of Hermitian Gell-Mann operators and their generalizations. Nevertheless, it is possible to construct a more compact and intuitive picture if we recall the generalization of the Bloch sphere representation in terms of the so-called Weyl operators, $W_{m,n}$~\cite{bertlmann2008bloch,kibler_2008_variations,siewert2022orthogonal}. Under this framework, an arbitrary density matrix $\varrho$ associated with a $D$-dimensional quantum system may be expressed as
\begin{equation}
\varrho = \frac{1}{D}\sum_{m,n=0}^{D-1}\beta_{m,n} ~W_{m,n},
\label{WeylRep}
\end{equation}
where the state coefficients are in general complex and satisfy the relations $\beta_{m,n}^*=\omega_D^{-mn}\beta_{-m,-n}$ and $\beta_{0,0}=1$ so that the density matrix $\varrho$ satisfies the Hermiticity and unit-trace conditions, respectively. Moreover, the operators $W_{m,n}$ are unitary, constitute a basis for the Hilbert--Schmidt space, and can be written in a compact form as
\begin{equation}
W_{m,n}=\sum_{j=0}^{D-1}\omega_D^{jm}\dyad{j}{j\oplus n},
\end{equation}
where the symbol $\oplus$ indicates sum modulo $D$. Moreover, note that for the two-level case ($D$=2), the set of four Weyl operators is equivalent to the Pauli matrices plus the identity.

Despite the fact that the $W_{m,n}$ operators are not Hermitian, the Weyl representation (Eq. \ref{WeylRep}) has certain advantages. For example, in the specific context of QST, the grand sum of an arbitrary state $\varrho$ (Eq. \ref{GS}) is simply reduced to
\begin{equation}
\Sigma(\varrho)=\sum_{k,l=0}^{D-1}\bra{k}\varrho\ket{l} =\frac{1}{D}\sum_{jklmn}\omega^{jm}\beta_{m,n}\delta_{j,k}\delta_{l,j\oplus n}=\sum_{n=0}^{D-1}\beta_{0,n}=1+\sum_{n=1}^{D-1}\beta_{0,n}.
\end{equation}

Hence, $\Sigma(\varrho)$ depends only on the coefficients associated with the group of Weyl operators that generalizes the Pauli $\sigma_x$ matrix; in other words, the \textit{flip} operators $W_{0,n}=\sum_j\dyad{j}{j\oplus n}$. In this way, we can define domains in the space of states whose density matrices possess the same amount of texture. 
%For instance, in the Bloch sphere representation, the grand sum associated to any state on the plane defined by a fixed $x$ component of the Bloch vector, attains the same value $\Sigma$ [recall \Eref{eq:qubit-GS}]. Example provided in the next paragraph:
These regions will be very helpful in the characterization of the free operations. 

In addition, for the case of qubits ($D=2$), \Eref{WeylRep} is equivalent to the Bloch sphere representation [\Eref{eq:qubit-bloch}] and the grand sum reduces to $\Sigma(\varrho)=1+x$ [\Eref{eq:qubit-GS}]. Thus, any quantum channel whose action on an arbitrary density matrix $\varrho$ is to displace its Bloch vector within the circular plane defined by the normal vector $(1,0,0)$, passing through the point $(x,0,0)$ is a free operation that \textit{preserves} the amount of texture of the quantum state. On the other hand, quantum channels whose effect results in an increase in $x$ are \textit{texture-decreasing} free operations. Finally, notice that channels decreasing the value of $x$, although physically valid, do not correspond to free operations in the resource theory of quantum texture.

\subsection{Texture-preserving operations}
According to the previous observation, any quantum channel $\mathcal{E}$,
\begin{equation}
\varrho \to \varrho'=\mathcal{E}[\varrho]=\frac{1}{D}\sum_{m,n=0}^{D-1}\beta'_{m,n} ~W_{m,n},
\end{equation}
that is expected to preserve the amount of texture in a quantum state must satisfy the general condition
\begin{equation}
\label{Eq_Cond_Pres_Op_Gen}
\sum_{n=0}^{D-1}\beta'_{0,n}=\sum_{n=0}^{D-1}\beta_{0,n}.
\end{equation}
for any density matrix $\varrho$.
Nevertheless, in the examples presented below, the constraints on the Weyl coefficients are more restrictive, as they are required to remain invariant:
\begin{equation}
\label{Eq_Cond_Pres_Op}
\beta'_{0,n}=\beta_{0,n}, ~~~~~ n=0,\cdots D-1.
\end{equation}

Among the physically admissible texture preserving quantum operations, it is possible to recognize two important instances involving specific classes of unitary operations and unital channels.
\subsubsection{Unitary Operations}
Recalling, the grand sum of an arbitrary density matrix $\varrho$ may be written alternatively as $\Sigma(\varrho)=D\bra{f_1}\varrho\ket{f_1}$. Thus, any $D\times D$ unitary $U$,
% %
% \[
% U = e^{i\delta} \begin{pmatrix}
% 1 & 0~  \dots ~ 0 \\
% 0 &  \\
% \vdots & \Tilde{U}_{D-1,D-1} \\
% 0 &
% \end{pmatrix}_f,
% \]
%
$$U = e^{i\delta}
\left(
\begin{array}{c|ccc}
  1 & 0 &  \cdots & 0  \\
  \hline
  0 &  & \\
  \vdots &  & \Tilde{U}_{D-1,D-1}\\
  0 &
\end{array}
\right)_f
$$
written in the Fourier basis $\left\{\ket{f_1},\ket{f_2},\cdots,\ket{f_{D}}\right\}$, is a free operation, where $\Tilde{U}_{D-1,D-1}$, is a unitary, acting only on the subspace of the Hilbert space spanned by $\left\{\ket{f_2},\cdots,\ket{f_{D}}\right\}$, with $\delta$ an arbitrary real phase. For the simplest case $(D=2)$, and after some calculations it is possible to show that any unitary in the set of allowed operations may be written in the standard computational basis as

\[
U = e^{i\phi} \begin{pmatrix}
\cos\varphi/2& i\sin\varphi/2\\
i\sin\varphi/2 & \cos\varphi/2
\end{pmatrix}_c.
\]

In other words, arbitrary rotations around the $x$ axis of the Bloch sphere, as expected, where $\phi$ and $\varphi$ are real phases.

% \[
% \hat{U} = e^{i\delta} \begin{pmatrix}
% 1 & 0\\
% 0 &  e^{-i\varphi}
% \end{pmatrix}_f
% \]
% %
% $$\hat{U}=e^{i\varphi} \left(e^{i\delta}\dyad{f_1} + e^{-i\delta}\dyad{f_2}\right)$$
%

% given a Bloch vector $\vec{r}=(r_x,r_y,r_z)$,

% For $d=2$, we have
% $$d=2$$\\

\subsection{Unital Channels}

Now let us introduce two families of texture-preserving unital quantum channels i.e. CPTP maps that leave the maximally mixed state invariant $\mathcal{E}\left[\mathds{1}_{D}/D\right]=\mathds{1}_{D}/D$. Usually, these kinds of channels are employed to model the effect of noise on quantum systems and non dissipative interactions with the environment \cite{Nielsen2010}.

\subsubsection{Weyl Channels}

Consider a physical process acting on a quantum system with density matrix $\varrho $. This process can be effectively characterized by the attenuation and/or phase shifting of its Weyl coefficients: $\beta_{m,n} \to \beta'_{m,n} = \tau_{m,n}\beta_{m,n}$, i.e.
\begin{equation}
\label{Cond_Weyl_Chann}
\varrho \to \varrho' = \frac{1}{D}\sum_{m,n=0}^{D-1} \tau_{m,n}  \beta_{m,n} W_{m,n}.
\end{equation}
For more details, constraints on the coefficients\footnote{Trace and Hermiticity preservation imply $\tau_{0,0}=1$ and $\tau_{m,n}=\tau^*_{-m,-n}$, respectively. In addition, the CPTP condition of the map implies that the following $D^2$ inequalities must be fulfilled: $\sum_{m,n=0}^{D-1}\tau_{m,n}~ \omega_D^{kn-lm}\geq 0$, for $k,l = 0,\ldots,D-1$ \cite{Basile2024}.} $\tau_{m,n}$, and a characterization of the multipartite case, we refer the interested reader to Ref.~\cite{Basile2024}.

The grand sum for a density matrix after the action of a Weyl channel reads
\begin{equation}
\Sigma(\varrho')=\sum_{n=0}^{D-1}\tau_{0,n}\beta_{0,n}=1+\sum_{n=1}^{D-1}\tau_{0,n}\beta_{0,n}.
\end{equation}
Hence, apart from the trivial case (the identity channel, e.g. $\tau_{m,n}=1$), there exists another map in agreement with the constraint in Eq. \ref{Eq_Cond_Pres_Op}. It is a channel which erases all Weyl coefficients but those associated to the operators $W_{0,n}$, for $n=0,\cdots,D-1$, in other words $\tau_{m,n}=\delta_{m,0}$. For a qubit, this is a map that projects any quantum state to the $x$ axis in the Bloch sphere.

\subsubsection{Flip channels}
Another interesting example of operations which preserve the amount of texture in a quantum state is the generalization of the bit flip channel. For the simplest case\footnote{For the sake of simplicity, we consider the probability of the system to jump one level up is equal to jump two, and so on.}, consider a $D$-level system, prepared in the state $\ket{k}$. Thus, the effect of this channel is to flip the state either to $\ket{k\oplus 1}$, $\ket{k\oplus 2}$, $\cdots$, or $\ket{k\oplus D-1}$, with probability $p/(D-1)$, each. This channel can be written in terms of Kraus operators, $K_m$ as \cite{Fonseca2019}:
\begin{equation}
\varrho \to \varrho' = \sum_{m=0}^{D-1} K_m \varrho K_m^{\dagger},
\end{equation}
where
\begin{equation}
K_0=\sqrt{1-p} ~ \mathds{1}, ~~~ K_m=\sqrt{\frac{p}{D-1}}W_{0,m},
~~~ m=1, \cdots D-1.
\end{equation}
It is possible to write this channel in the form of Eq. \ref{Cond_Weyl_Chann}, in other words, finding the coefficients $\tau_{m,n}$. After some calculations, we have
\begin{equation}
\tau_{0,n}=1, ~~~ n=0,\cdots, D-1,
\end{equation}
and
\begin{equation}
\tau_{j,n}=\frac{D(1-p)-1}{D-1}, ~~~ j=1,\cdots, D-1.
\end{equation}
Note that as in the previous example, given that $\tau_{0,n}=1$, then the grand sum remains invariant, as expected. However, it is important to note that in this case, the other coefficients not necessarily vanish.

% $$\varrho \to \varrho' = (1-p)\varrho + p~ \hat{\sigma}_x \varrho \hat{\sigma}_x$$

\subsection{Texture decreasing operations } % {{{* 
In this subsection, we focus our attention on quantum operations $\Lambda$, that diminish the amount of texture (or increase the grand sum) for any initial state. Recalling the expression for the grand sum in terms of the Weyl coefficients [\Eref{Eq_Cond_Pres_Op_Gen}], we have
\begin{equation}
\label{Eq_Cond_Dec_Op_Gen}
\Sigma(\varrho')=\sum_{n=0}^{D-1}\beta'_{0,n}\geq\sum_{n=0}^{D-1}\beta_{0,n}=\Sigma(\varrho),
\end{equation}
for any density matrix $\varrho$. Note that for a texture decreasing operation $\Lambda$, the equality is attained only when the density matrix $\varrho ~(=\varrho')$ corresponds to the textureless state $\dyad{f_1}$, for any valid texture decreasing operation $\Lambda$. In other words, we observe that, as expected, the textureless state is a fixed point for any map within the set.

\subsubsection{Amplitude damping channels}
A paradigmatic example of quantum map with a fixed point is the amplitude damping channel, which is usually employed to model dissipative processes, spontaneous emission, among others in two level quantum systems \cite{Nielsen2010}. High dimensional generalizations have been proposed \cite{Dutta16}, and more recently, these models have been generalized to cover cases in which the system presents unequal transition probabilities between the levels \cite{PhysRevA.102.012401,Chessa2023resonantmultilevel}.

Motivated by these generalizations of amplitude damping channels, it is possible to introduce a class of operations which reduce the amount of texture in any quantum state, if we employ the textureless state $\ket{f_1}$, as the fixed point. In this way, it is not difficult to note that for the simplest case (equal decaying probabilities, $p$), the associated Kraus operators may be written as
%
% $$\varrho \to \varrho' = \sum_{m=1}^{d} \hat{K}_m \varrho \hat{K}_m^{\dagger}$$
%
\begin{equation}
 \hat{K}_1=\dyad{f_1}+\sqrt{1-p}\sum_{m=2}^D\dyad{f_m}
\end{equation}
and
\begin{equation}
\hat{K}_m=\sqrt{p}\dyad{f_1}{f_m}, ~~~~~ m=2,\cdots,D.
\end{equation}

After some steps, the grand sum is reduced to
\begin{equation}
\Sigma(\varrho')=Dp+\left(1-p\right)\Sigma(\varrho).
\end{equation}
Therefore, this is a map that increases the grand sum (and consequently decreases the texture), attaining its maximum for $p=1$, $\Sigma(\varrho')=D$, as expected. In addition, note that the latter constitutes an example of a texture destroying operation.

\subsubsection{Generalized amplitude damping channels}
A more realistic situation is one in which the decay probabilities are not equal. For this, the Kraus operators read
\begin{equation}
 \hat{K}_1=\dyad{f_1}+\sum_{m=2}^D\sqrt{1-p_m}\dyad{f_m}
\end{equation}
and
\begin{equation}
\hat{K}_m=\sqrt{p_m}\dyad{f_1}{f_m}, ~~~~~ m=2,\cdots,D,
\end{equation}
where $p_m$ is the probability of the system to decay from $\ket{f_m}$ to the target state $\ket{f_1}$.
In this case, we have
\begin{equation}
\Sigma(\varrho')=\Sigma(\varrho)+D\sum_{m=2}^Dp_m\expval{\varrho}{f_m},
\end{equation}
where we observe, once again, an increase in the value of the grand sum.

% *}}}
% *}}}
\section{Bipartite systems} % {{{* 6
\label{sec:bipartite-systems}
% Intro{{{*
Taking the partial trace over a state of a composite system is arguably the most relevant coarse graining map. Let us investigate the behavior of QST, or, equivalently, the grand sum of reduced density operators in bipartite systems of dimension $D=D_1D_2$, where $D_1$ and $D_2$ are the dimensions of the subsystems. In subsection.~\ref{sec:bounds_global}, we start by establishing bounds involving the grand sums of the full system state and the grand sums of the reduced systems, for a fixed basis. In subsection \ref{sec:bounds_distillable} we find  bounds involving the minimum and maximum values attained by the distillable texture across all bases, which enables us to determine an inequality satisfied by a new purity measure, defined in Ref.~\cite{aditi}.

% *}}}
\subsection{Bounds on the global and reduced grand sums}\label{sec:bounds_global} % {{{*
We begin by deriving upper and lower bounds for $\Sigma(\varrho)$ in terms of the grand sums of its reduced states, which refine the trivial range $0\leq\Sigma(\varrho)\leq D=D_1D_2$.
We start by rewriting the grand sum of an arbitrary state $\varrho$ as
\begin{equation}\label{eq:GS_bipartite}
  \Sigma(\varrho) = 
    D_1 D_2 \Tr(\varrho \dyad{f_1^{(D_1)}} \otimes \dyad{f_1^{(D_2)}}).
\end{equation}
Similarly, for the first subsystem:
\begin{equation}\label{eq:GS_subsystem1}
\begin{split}
  \Sigma(\varrho_1) &= D_1 \Tr(\varrho_1 \dyad{f_1^{(D_1)}})= D_1 \Tr(\varrho \dyad{f_1^{(D_1)}} \otimes \mathbb{I}_{D_2}),
\end{split}
\end{equation}
where $\varrho_1$ denotes the reduced density matrix of the first subsystem and $\mathbb{I}_{D_2}$ the identity operator of dimension $D_2$ acting over the second subsystem.

% {\color{red} % Old {{{*
% Next, we multiply the reduced grand sum in \Eref{eq:GS_subsystem1} by $D_2$, obtaining $D_2 \Sigma(\varrho_1) = D_1 D_2 \Tr\qty(\varrho \dyad{f_1^{D_1}} \otimes \mathbb{I}_{D_2})$. We then take the difference between this and the global grand sum in \Eref{eq:GS_bipartite}:
% %
% \begin{equation}
% \begin{split}
%   D_2 \Sigma(\varrho_1) - \Sigma(\varrho) 
%     &= D_1 D_2 \Tr(\varrho \dyad{f_1^{(D_1)}} \otimes \mathbb{I}) - D_1 D_2 \Tr(\varrho \dyad{f_1^{(D_1)}} \otimes \dyad{f_1^{(D_2)}}) \\
%   %    
%     &= D_1 D_2 \Tr \bigg[\varrho \dyad{f_1^{(D_1)}} 
%      \otimes \pqty{\mathbb{I} - \dyad{f_1^{(D_2)}}} \bigg].
% \end{split}
% \end{equation}
% %
% Both operators $\dyad{f_1^{(D_1)}}$ and $\pqty{\mathbb{I} - \dyad{f_1^{(D_2)}}}$ are projectors.  
% The trace of the product of two positive semi-definite matrices is always non-negative. Consequently, $ D_2 \Sigma(\varrho_1) - \Sigma(\varrho) \geq 0$, or $\Sigma(\varrho)\le D_2 \Sigma(\varrho_1)$.
% Of course the second relation, $ \Sigma(\varrho)\le D_1 \Sigma(\varrho_2)$
% is also valid. Therefore, we have
% %
% \begin{equation}\label{eq:ineq:2}
% \Sigma(\varrho)\le {\rm min}\left\{D_2 \Sigma(\varrho_1),D_1 \Sigma(\varrho_2)\right\}.
% \end{equation}
% } % *}}}

An upper bound for the grand sum in \Eref{eq:GS_bipartite} can be obtained as follows. First, consider the projector onto the tensor product of the textureless subspace of the first subsystem and the subspace orthogonal to the textureless state of the second subsystem. Since both $\varrho$ and this projector are positive semidefinite,
\begin{equation}
\begin{split}
  0 &\leq \Tr\qty[
  \varrho\dyad{f_1^{(D_1)}}\otimes
  \qty(\mathbb{I}-\dyad{f_1^{(D_2)}})
  ] \\
  &= \frac{\Sigma(\varrho_1)}{D_1}
    -\frac{\Sigma(\varrho)}{D}.
\end{split}
\end{equation}
This implies
\begin{equation}
  \Sigma(\varrho)\leq D_2\Sigma(\varrho_1).
\end{equation}
Second, exchanging the roles of the two subsystems gives
$\Sigma(\varrho)\leq D_1\Sigma(\varrho_2)$. Combining these two inequalities, we obtain the upper bound
\begin{equation}\label{eq:ineq:2}
\Sigma(\varrho)\le {\rm min}\left\{D_2 \Sigma(\varrho_1),D_1 \Sigma(\varrho_2)\right\}.
\end{equation}
In terms of rugosity, see \Eref{rug}, this becomes $\mathfrak{R}(\varrho)\ge {\rm max}\left\{\mathfrak{R}(\varrho_1),\mathfrak{R}(\varrho_2)\right\}$.
Since the only assumption made about the global state $\varrho$ is its positive semi-definiteness, this constraint is completely general. An interesting consequence of (\ref{eq:ineq:2}) is that there is no global state with non-zero grand sum and one of the reduced systems with a state in the orthogonal support of $f_1$. 

The inequality in \Eref{eq:ineq:2} naturally leads one to consider whether the stronger relation $\Sigma(\varrho)\le^{?} \Sigma(\varrho_1)\Sigma(\varrho_2)$ holds. 
% Since $\Sigma(\varrho_1)\Sigma(\varrho_2)\leq \min\left\{D_2\Sigma(\varrho_1),D_1\Sigma(\varrho_2)\right\}$, \Eref{eq:ineq:2} naturally leads one to ask whether the stronger relation $\Sigma(\varrho)\le^{?}\Sigma(\varrho_1)\Sigma(\varrho_2)$ holds.
The answer is negative, as the following counterexample shows. Consider $|\psi\rangle=a|00\rangle+b|11\rangle$, with $a$ real and non-negative. We readily obtain $\Sigma(\varrho)-\Sigma(\varrho_1)\Sigma(\varrho_2)=2a{\rm Re}(b)$, which can be positive, negative, or zero.

% The meaning of these results become even clearer when expressed in terms of rugosity. 
% %
% \begin{equation}\label{eq:ineq:3}
% \mathfrak{R}(\varrho)\ge {\rm max}\left\{\mathfrak{R}(\varrho_1),\mathfrak{R}(\varrho_2)\right\}.
% \end{equation}

A lower bound for the grand sum in \Eref{eq:GS_bipartite} can be obtained similarly. First, the positive semidefiniteness of $\varrho$ implies that $\Sigma(\varrho)\geq 0$. Second, consider the projector onto the tensor product of the local subspaces orthogonal to the textureless states. Since both $\varrho$ and this projector are positive semidefinite,
\begin{equation}
\begin{split}
  0 &\leq \Tr\qty[
  \varrho\qty(\mathbb{I}-\dyad{f_1^{(D_1)}})\otimes
  \qty(\mathbb{I}-\dyad{f_1^{(D_2)}})
  ] \\
  &= 1-\frac{\Sigma(\varrho_1)}{D_1}
    -\frac{\Sigma(\varrho_2)}{D_2}
    +\frac{\Sigma(\varrho)}{D}.
\end{split}
\end{equation}
This implies
\begin{equation}
  \Sigma(\varrho)
  \geq D_2\Sigma(\varrho_1)
    + D_1\Sigma(\varrho_2)
    - D.
\end{equation}
Combining these two inequalities, we obtain the lower bound
\begin{equation}
  \label{eq:ineq:lower}
  \max\qty{
    0,\,
    D_2\Sigma(\varrho_1)
    + D_1\Sigma(\varrho_2)
    - D
  } \leq \Sigma(\varrho).
\end{equation}

Equations \eqref{eq:ineq:2} and \eqref{eq:ineq:lower} can also be used to obtain bounds in both directions: bounds on the product of the subsystem grand sums at fixed global grand sum and, conversely, bounds on the global grand sum at fixed product. We begin with the former. Multiplying the two inequalities $\Sigma(\varrho)\leq D_2\Sigma(\varrho_1)$ and $\Sigma(\varrho)\leq D_1\Sigma(\varrho_2)$, found in \Eref{eq:ineq:2}, gives the lower bound
\begin{equation}
\label{eq:lower_bound}
  \Sigma(\varrho_1)\Sigma(\varrho_2)
  \geq
  \frac{\Sigma(\varrho)^2}{D}.
\end{equation}
An upper bound on the same product follows directly from \Eref{eq:ineq:lower}. Before taking the maximum with zero, the inequality can be written as
\begin{equation}
  D_2\Sigma(\varrho_1)
  +
  D_1\Sigma(\varrho_2)
  \leq
  D+\Sigma(\varrho).
\end{equation}
By the arithmetic--geometric mean inequality,
\begin{equation}
  2\sqrt{
    D_2\Sigma(\varrho_1)
    D_1\Sigma(\varrho_2)
  }
  \leq
  D_2\Sigma(\varrho_1)
  +
  D_1\Sigma(\varrho_2)
  \leq
  D+\Sigma(\varrho).
\end{equation}
Squaring both sides gives the upper bound
\begin{equation}
\label{eq:upper_bound}
  \Sigma(\varrho_1)\Sigma(\varrho_2)
  \leq
  \frac{
    \qty[D+\Sigma(\varrho)]^2
  }{4D}.
\end{equation}
Consequently, the exact window for the product of the subsystem grand sums at fixed global grand sum is
\begin{equation}
\label{eq:product_bounds}
  \frac{\Sigma(\varrho)^2}{D}
  \leq
  \Sigma(\varrho_1)\Sigma(\varrho_2)
  \leq
  \frac{
    \qty[D+\Sigma(\varrho)]^2
  }{4D}.
\end{equation}
Conversely, these inequalities may be expressed as bounds on the global grand sum:
\begin{equation}
\label{eq:global_product_bounds}
  \max\qty{
    0,\,
    2\sqrt{
      D\Sigma(\varrho_1)\Sigma(\varrho_2)
    }
    -
    D
  }
  \leq
  \Sigma(\varrho)
  \leq
  \sqrt{
    D\Sigma(\varrho_1)\Sigma(\varrho_2)
  }.
\end{equation}
Therefore, \Eref{eq:product_bounds} and \Eref{eq:global_product_bounds} provide two equivalent ways of expressing the constraints derived from \Eref{eq:ineq:2} and \Eref{eq:ineq:lower}.

For $D_1,D_2\geq2$, the bounds in \Eref{eq:product_bounds} are tight. To exhibit saturating families, let $\varrho_\perp^{(D_j)}$ be a state satisfying
\begin{equation}
  \mel{f_1^{(D_j)}}{\varrho_\perp^{(D_j)}}{f_1^{(D_j)}}=0,
  \qquad j=1,2.
\end{equation}
For any $p\in[0,1]$, consider
\begin{equation}
  \varrho_{\mathrm{min}}(p)
  =
  p\,\dyad{f_1^{(D_1)}}\otimes\dyad{f_1^{(D_2)}}
  +(1-p)\,\varrho_\perp^{(D_1)}\otimes\varrho_\perp^{(D_2)}.
\end{equation}
We denote the corresponding reduced states by $\varrho_{\mathrm{min},1}(p)=\Tr_2[\varrho_{\mathrm{min}}(p)]$ and $\varrho_{\mathrm{min},2}(p)=\Tr_1[\varrho_{\mathrm{min}}(p)]$. For this family,
\begin{equation}
  \Sigma\qty(\varrho_{\mathrm{min}})=pD,
  \qquad
  \Sigma\qty(\varrho_{\mathrm{min},1})=pD_1,
  \qquad
  \Sigma\qty(\varrho_{\mathrm{min},2})=pD_2,
\end{equation}
and hence
\begin{equation}
  \Sigma\qty(\varrho_{\mathrm{min},1})
  \Sigma\qty(\varrho_{\mathrm{min},2})
  =
  \frac{\Sigma\qty(\varrho_{\mathrm{min}})^2}{D}.
\end{equation}
Therefore, this family saturates the lower product bound.

A family saturating the upper product bound is
\begin{equation}
\begin{split}
  \varrho_{\mathrm{max}}(p)
  ={}&
  p\,\dyad{f_1^{(D_1)}}\otimes\dyad{f_1^{(D_2)}}
  +\frac{1-p}{2}
  \qty(
    \dyad{f_1^{(D_1)}}\otimes\varrho_\perp^{(D_2)}
  )
  \\
  &+
  \frac{1-p}{2}
  \qty(
    \varrho_\perp^{(D_1)}\otimes\dyad{f_1^{(D_2)}}
  ).
\end{split}
\end{equation}
Similarly, we denote its reduced states by
$\varrho_{\mathrm{max},1}(p)=\Tr_2[\varrho_{\mathrm{max}}(p)]$ and
$\varrho_{\mathrm{max},2}(p)=\Tr_1[\varrho_{\mathrm{max}}(p)]$.
Its global grand sum is
\begin{equation}
  \Sigma\qty(\varrho_{\mathrm{max}})=pD,
\end{equation}
whereas the reduced grand sums are
\begin{equation}
  \Sigma\qty(\varrho_{\mathrm{max},1})=\frac{1+p}{2}D_1,
  \qquad
  \Sigma\qty(\varrho_{\mathrm{max},2})=\frac{1+p}{2}D_2.
\end{equation}
It follows that
\begin{equation}
  \Sigma\qty(\varrho_{\mathrm{max},1})
  \Sigma\qty(\varrho_{\mathrm{max},2})
  =
  \frac{D(1+p)^2}{4}
  =
  \frac{\qty[D+\Sigma\qty(\varrho_{\mathrm{max}})]^2}{4D},
\end{equation}
so the upper product bound is saturated for every $p\in[0,1]$. Together, these tight inequalities provide complementary descriptions of how the global grand sum constrains, and is constrained by, the grand sums of the two subsystems.

% *}}}
\subsection{Basis-independent bounds on the global distillable texture} % {{{*
\label{sec:bounds_distillable}
Consider the distillable texture $T(\varrho)=1-\langle f_1|\varrho|f_1\rangle$, defined in Sec.~\ref{sec:preliminaries}.
Recently, an elegant result on the maximal and minimal values assumed by $T(\varrho)$, across all possible orthogonal bases, has been derived in Ref.~\cite{aditi}, and reads (in a slightly distinct notation):
\begin{equation}
T_{\rm min}=1-\lambda^{\uparrow}_D, \;\;\; T_{\rm max}=1-\lambda^{\uparrow}_{1},
\end{equation}
where $\{\lambda^{\uparrow}_1, \lambda^{\uparrow}_2, \dots, \lambda^{\uparrow}_D\}$ stands for the eigenvalues of $\varrho$ in non-decreasing order. These basis-independent relations have enabled interesting connections between QST and other resources \cite{aditi}.
Here we note that, complementarily to those results,  there are general inequalities valid for the eigenvalues of $\varrho$ and the eigenvalues of the reduced density matrices \cite{eigenvalues}. In the bipartite case these relations are
\begin{equation}
\sum_{i=1}^{k}\lambda^{\uparrow}_{i}(\varrho_1)+\sum_{j=1}^{\ell}\lambda^{\uparrow}_{j}(\varrho_2)\ge \sum_{i=1}^{kD_2+\ell D_1-k\ell}\lambda^{\uparrow}_{i}(\varrho)+\sum_{j=1}^{k\ell}\lambda^{\uparrow}_{j}(\varrho),
\end{equation}
where $k=1, \dots, D_1-1$ and $\ell=1, \dots, D_2-1$. In particular, by setting $k=\ell=1$ one obtains a result that can be expressed in terms of extremal texture values as follows
\begin{equation}
T_{\rm max}(\varrho)\ge\frac{T_{\rm max}(\varrho_1)+T_{\rm max}(\varrho_2)}{2},
\label{ineqmax}
\end{equation}
which simply states that the maximum distillable texture of the global state cannot be smaller than the average of the maximum distillable textures of the reduced systems (note that these maximum values are generally attained for \textit{different} bases).
Analogously, for $k=D_1-1$ and $\ell=D_2-1$, we get
\begin{equation}
T_{\rm max}(\varrho)- T_{\rm min}(\varrho) \ge 1-T_{\rm min}(\varrho_1)-T_{\rm min}(\varrho_2).
\label{ineqmin}
\end{equation}
Considering that the difference 
\begin{equation}
\mathfrak{P}(\varrho)\equiv T_{\rm max}(\varrho)- T_{\rm min}(\varrho)
\end{equation} 
has been shown to be a purity monotone in Ref.~\cite{aditi}, we see that this purity is lower bounded by $1-T_{\rm min}(\varrho_1)-T_{\rm min}(\varrho_2)$. Therefore, one can write \Eref{ineqmin} as:
\begin{equation}
 \mathfrak{P}(\varrho)\ge\frac{\mathfrak{P}(\varrho_1)+\mathfrak{P}(\varrho_2)}{2}+ 1-\frac{ T_{\rm max}(\varrho_1)+T_{\rm max}(\varrho_2)+T_{\rm min}(\varrho_1)+T_{\rm min}(\varrho_2)}{2}.
 \end{equation}
Now, we note that $\lambda^{\uparrow}_D+\lambda^{\uparrow}_{1}\le 1$, since the sum of all eigenvalues must be equal to 1, and thus $T_{\rm max}(\varrho_j)+T_{\rm min}(\varrho_j)\ge 1$. So, a weaker, but perhaps more informative form of the previous inequality is simply
\begin{equation}
 \mathfrak{P}(\varrho)\ge\frac{\mathfrak{P}(\varrho_1)+\mathfrak{P}(\varrho_2)}{2}.
 \label{averageP}
\end{equation}
Interestingly, there is no equivalent relation for the standard purity $P(\varrho)={\rm Tr}(\varrho^2)$. Consider, for instance, $\varrho = \mathds{1}_4/4$ and $\varrho_1 =\varrho_2= \mathds{1}_2/2$ for which we have $P(\varrho) = 1/4$ and $P(\varrho_1)=P(\varrho_2)=1/2$. On the other hand $\mathfrak{P}(\mathds{1}_4/4)=\mathfrak{P}(\mathds{1}_2/2)=0$, satisfying (and saturating) the inequality in \Eref{averageP}.

%

% *}}}
% *}}}
\section{Set texture} % {{{* 7
\label{sec:set-texture}
% Intro {{{*

A common criticism to standard coherence theory is that the amount of coherence assigned to a state depends on the reference basis. This motivated several attempts to formulate basis-independent notions of coherence. A natural approach is to minimize the coherence over all bases, for any individual state, leading to ``total coherence'' or related basis-independent quantities \cite{Yu2016,Radhakrishnan2019,Jin2024}. These constructions are informative, but they effectively measure the departure from maximal mixedness rather than incompatibility with a common classical description.

The notion of set coherence approaches the problem from a different, but closely related perspective. It is a resource associated with \emph{families} of states. A set $\cS=\{\varrho_1,\dots,\varrho_N\}$ is set-incoherent if there exists a single orthonormal basis in which all members of the set are simultaneously diagonalizable~\cite{Designolle2021}. For density operators this condition is equivalent to pairwise commutativity
$[\varrho_a,\varrho_b]=0\quad \forall a,b$ \cite{Designolle2021}.
%Set coherence is then quantifies the deviation from this free structure by minimizing basis-dependent coherence quantities over one common basis \cite{Designolle2021}.

Inspired by the idea of set coherence, this section introduces the concept of set texture, an extension of quantum-state texture from a single density operator to a family of states. The motivation for this construction originates from the basis dependence of standard QST. Given any set of states, one can ask what is the  texture of the whole set in a fixed basis. For example, a simple way to define this overall texture is to sum up the rugosity (an additive measure) of each state in the set. The resulting number is still basis dependent, but one can ask, out of all possible bases, which one gives the smaller amount of QST. This number can be seen as the basis-independent essential QST attached to the set. 
%By requiring a single reference basis for the entire family, set texture preserves information about whether the states admit compatible texture-minimizing descriptions. 
This common-basis requirement is also natural whenever several possible preparations must be described or manipulated within the same reference frame.

The optimization problem is then to minimize (maximize) the overall texture (grand sum) contributions of the states over all bases. We remark that the resulting set texture depends on how one aggregates the states, e.g., via tensor product or direct sum. From these two perspectives, we define arithmetic and geometric set textures, as well as natural extensions stemming from these two perspectives.  In addition, we evaluate these quantifiers for sets of representative states. These examples illustrate the limiting cases of the proposed measures and the different information retained by the different aggregation rules.

% *}}}
%\subsection{Set coherence} % {{{*

% *}}}
\subsection{Setup and notation} % {{{*

Let $\cS=\{\varrho_1,\varrho_2,\dots,\varrho_N\}$ be a family of density operators on the same $D$-dimensional Hilbert space. For a unitary $U$, define
\begin{equation}
\Sigma_U(\varrho_a):= D \bra{f_1}U\varrho_a U^{\dagger}\ket{f_1},
\end{equation}
where $\Sigma_U(\varrho_a)$ denotes the grand sum of $\varrho_a$ in the basis selected by $U$. Since $U\varrho_a U^{\dagger}$ is again a density matrix, each $\Sigma_U(\varrho_a)$ is real and satisfies
\begin{equation}
0\le \Sigma_U(\varrho_a)\le D.
\end{equation}
The associated single-state rugosity in the basis selected by $U$ is
\begin{equation}
\mathfrak{R}_U(\varrho_a):=-\ln \frac{\Sigma_U(\varrho_a)}{D}.
\end{equation}

Here, the grand sum $\Sigma_U(\varrho_a)$ measures how well the state $\varrho_a$ can be aligned with the textureless direction under a common basis choice $U$. The remaining question is how one should aggregate the $N$ overlaps $\{\Sigma_U(\varrho_a)\}_{a=1}^N$ into a single family-level quantity. Different aggregation principles lead to different notions of set rugosity. Below we show that two particularly natural choices arise from composing the family members through a direct sum and via tensoring.  A third, more conservative framework, then emerges from the generalized-mean hierarchy as the worst-case limit.

% *}}}
\subsection{Tensor-product framework and geometric rugosity} % {{{*

To define a family-level texture quantifier in the spirit of set coherence, we begin by asking how a whole set of states may be represented relative to one common textureless direction. A natural first possibility is to treat the family as a system composed by the tensor product of the states.

Consider therefore $\varrho_\otimes:=\varrho_1\otimes\cdots\otimes\varrho_N$ on $\mathcal H_\otimes:=\mathcal H^{\otimes N}$. If the same basis change is applied to every factor, then
\begin{equation}
U_\otimes:=U^{\otimes N},
\qquad
\ket{F_\otimes}:=\ket{f_1}^{\otimes N}.
\end{equation}
The corresponding overlap is
\begin{equation}
\bra{F_\otimes}U_\otimes \varrho_\otimes U_\otimes^\dagger \ket{F_\otimes}
=
\prod_{a=1}^N \frac{\Sigma_U(\varrho_a)}{D}.
\end{equation}

Thus, within the tensor-product framework, the natural family overlap is multiplicative rather than additive. Taking the logarithmic version of this quantity leads to the \emph{geometric rugosity}
\begin{equation}
\mathfrak{R}_G(\cS)
:=
-\ln\!\left[
\max_U \left(\prod_{a=1}^N \frac{\Sigma_U(\varrho_a)}{D}\right)^{1/N}
\right].
\label{eq:RGdef}
\end{equation}
Equivalently,
\begin{equation}
\mathfrak{R}_G(\mathcal S)
=
\frac{1}{N}
\left(
-\ln\!\left[
\max_U
\bra{F_\otimes}U_\otimes \varrho_\otimes U_\otimes^\dagger \ket{F_\otimes}
\right]
\right).
\end{equation}
%So we see that the tensor-product framework is naturally adapted to the geometric rugosity: the family contributions accumulate multiplicatively, and $\mathfrak{R}_G$ is the corresponding intensive, per-copy rugosity.

Moreover, the geometric form is the closest analogue of the original single-state rugosity, since it keeps the logarithm and combines family contributions in a natural multiplicative way. 

\begin{proposition}
We know that the maximizing unitary in Eq.~\eqref{eq:RGdef} is such that $\Sigma_U(\varrho_a)\ge0$ for all $a$. Then
\begin{equation}
\mathfrak{R}_G(\cS)=\min_U \frac{1}{N}\sum_{a=1}^N \mathfrak{R}_U(\varrho_a),
\label{eq:RGcompact}
\end{equation}
where $\mathfrak{R}_U(\varrho_a)=-\ln \frac{\Sigma_U(\varrho_a)}{D}.$

\end{proposition}

\begin{proof}
Define
\begin{equation}
F(U):=\left(\prod_{a=1}^N \frac{\Sigma_U(\varrho_a)}{D}\right)^{1/N}.
\end{equation}
Because each $\Sigma_U(\varrho_a)$ is real and nonnegative, $F(U)$ is real and nonnegative as well. On the region where $F(U)>0$, the function $-\ln x$ is strictly decreasing on $(0,\infty)$; therefore, maximizing $F(U)$ is equivalent to minimizing $-\ln F(U)$:
\begin{equation}
-\ln\!\left[\max_U F(U)\right]=\min_U\bigl[-\ln F(U)\bigr].
\end{equation}
Now,
\begin{align}
-\ln F(U)
=-\ln\!\left(\prod_{a=1}^N \frac{\Sigma_U(\varrho_a)}{D}\right)^{1/N}=-\frac{1}{N}\sum_{a=1}^N \ln \frac{\Sigma_U(\varrho_a)}{D} =\frac{1}{N}\sum_{a=1}^N \bigl[-\ln \frac{\Sigma_U(\varrho_a)}{D}\bigr] =\frac{1}{N}\sum_{a=1}^N \mathfrak{R}_U(\varrho_a).
\end{align}
Substituting this identity into the previous equation yields Eq.~\eqref{eq:RGcompact}.
\end{proof}

Equation~\eqref{eq:RGcompact} is the compact form of the geometric rugosity and explains why it is a particularly appealing candidate: the set quantity can be interpreted directly as the minimum \emph{average single-state rugosity} under one shared basis. However, the optimization must be handled for each specific set, which makes the tensor product framework more costly compared to a direct-sum framework.

% *}}}
\subsection{Direct-sum framework and arithmetic rugosity} % {{{*

Another natural possibility is to treat the family as a block-structured object, in which each state occupies its own sector and the common basis change acts identically on every sector. This leads to a direct-sum encoding of the states.

Concretely, consider $\varrho_\oplus:=\varrho_1\oplus\cdots\oplus\varrho_N$ acting on $\mathcal H_\oplus:=\mathcal H\oplus\cdots\oplus\mathcal H$. If the same basis change is applied to every block, the corresponding unitary is
\begin{equation}
U_\oplus:=U\oplus\cdots\oplus U,
\end{equation}
and the natural reference vector built from the single-state textureless direction is
\begin{equation}
\ket{F_\oplus}:=\frac{1}{\sqrt N}\bigl(\ket{f_1}\oplus\cdots\oplus\ket{f_1}\bigr).
\end{equation}
A direct computation then yields
\begin{equation}
\bra{F_\oplus}U_\oplus \varrho_\oplus U_\oplus^\dagger \ket{F_\oplus}
=
\frac{1}{N}\sum_{a=1}^N \frac{\Sigma_U(\varrho_a)}{D}.
\end{equation}

Thus, within the direct-sum framework, the natural family overlap is the arithmetic average of the individual overlaps. This motivates the arithmetic rugosity
\begin{equation}
\mathfrak{R}_A(\cS)
:=
-\ln\!\left[\max_U \frac{1}{N}\sum_{a=1}^N \frac{\Sigma_U(\varrho_a)}{D}\right].
\label{eq:RA}
\end{equation}
This quantity is the logarithmic family analogue naturally associated with the direct-sum encoding.

The same arithmetic structure may be expressed without passing to the direct-sum space. Defining the mean state
\begin{equation}
\bar\varrho:=\frac{1}{N}\sum_{a=1}^N \varrho_a,
\end{equation}
one finds
\begin{equation}
\frac{1}{N}\sum_{a=1}^N \frac{\Sigma_U(\varrho_a)}{D}=\bra{f_1}U\bar\varrho U^\dagger\ket{f_1},
\end{equation}
so that
\begin{equation}\label{RAlambda}
\mathfrak{R}_A(\cS)=-\ln\lambda^{\uparrow}(\bar\varrho).
\end{equation}
In this sense, the direct-sum and mean-state pictures are two equivalent arithmetic realizations of the same framework. Moreover, we can see from Eq.~\eqref{RAlambda} that the optimization can be done analytically, giving us an equation easier to handle than its geometrical counterpart.

% If one nevertheless wishes to keep a nonlogarithmic companion in closer analogy with Patra's single-state construction, one may define
% \begin{equation}
% T_{\mathrm{set}}(\cS)
% :=
% 1-\max_U \frac{1}{N}\sum_{a=1}^N \frac{\Sigma_U(\varrho_a)}{D}
% =
% 1-\lambda^{\uparrow}(\bar\varrho).
% \label{eq:Tset}
% \end{equation}
% We will use this quantity occasionally as the arithmetic texture counterpart of $\mathfrak{R}_A$, although the logarithmic quantity $\mathfrak{R}_A$ will be the primary family rugosity in what follows. Equivalently,
% \begin{equation}
% \mathfrak{R}_A(\cS)=-\ln\!\bigl(1-T_{\mathrm{set}}(\cS)\bigr).
% \end{equation}

% *}}}
\subsection{Generalized-mean hierarchy and minimum rugosity} % {{{*

The tensor-product and direct-sum formalisms single out the geometric and arithmetic constructions. A natural next question is whether these two quantities belong to a broader family of possible set rugosities, obtained by changing the way the overlaps $\{\Sigma_U(\varrho_a)\}_{a=1}^N$ are aggregated.

This is naturally answered by the generalized means. For fixed $U$, define
\begin{equation}
M_p(U):=
\left(
\frac{1}{N}\sum_{a=1}^N \frac{\Sigma_U(\varrho_a)}{D}^p
\right)^{1/p},
\qquad p\neq 0,
\end{equation}
and let
\begin{equation}
M_0(U):=\lim_{p\to 0}M_p(U)
=
\left(\prod_{a=1}^N \frac{\Sigma_U(\varrho_a)}{D}\right)^{1/N}.
\end{equation}
The associated family rugosity is
\begin{equation}
\mathfrak{R}_p(\mathcal S):=
-\ln\!\left[\max_U M_p(U)\right].
\end{equation}

In this language, the previously introduced quantities appear as distinguished members of the same hierarchy:
\begin{equation}
\mathfrak{R}_A(\mathcal S)=\mathfrak{R}_1(\mathcal S),
\qquad
\mathfrak{R}_G(\mathcal S)=\mathfrak{R}_0(\mathcal S).
\end{equation}
The minimum rugosity then emerges naturally as the extremal negative-$p$ limit,
\begin{equation}
\mathfrak{R}_{\min}(\mathcal S)=\lim_{p\to -\infty}\mathfrak{R}_p(\mathcal S)
=
-\ln\!\left[\max_U \min_{1\le a\le N} \frac{\Sigma_U(\varrho_a)}{D}\right].
\label{eq:Rmin}
\end{equation}
Unlike the previous two quantities, the minimum rugosity does not appear to be naturally singled out by a standard linear or multiplicative state-composition rule. Rather, it is governed by a different principle: it asks for the basis that optimizes the \emph{worst represented} member of the family. For this reason, $\mathfrak{R}_{\min}$ is best interpreted as the robust or adversarial framework of the theory.

Having identified the geometric and arithmetic constructions as natural consequences of two distinct composite encodings, it is natural to ask whether other aggregation principles may lead to additional family rugosities. The opposite endpoint is governed by the largest overlap in the family:
\begin{equation}
\mathfrak{R}_{\max}(\mathcal S):=
\lim_{p\to+\infty}\mathfrak{R}_p(\mathcal S) = -\ln\!\left[
\max_U \max_a \frac{\Sigma_U(\varrho_a)}{D}
\right].
\end{equation}

Since
\begin{equation}
\max_U \frac{\Sigma_U(\varrho_a)}{D}=\lambda^{\uparrow}(\varrho_a),
\end{equation}
it follows that
\begin{equation}
\mathfrak{R}_{\max}(\mathcal S)
=
-\ln\!\left[\max_a \lambda^{\uparrow}(\varrho_a)\right].
\end{equation}

{\it Mixed composite encodings.}
The tensor-product and direct-sum frameworks motivate the geometric and arithmetic constructions in a particularly clean way. One may then ask whether more complicated composite encodings, built from both tensor products and direct sums, lead to further members of the generalized-mean hierarchy. In general, mixed encodings typically produce \emph{nested} aggregations rather than a single generalized mean. For instance, tensor products of direct-sum blocks lead naturally to geometric combinations of arithmetic averages, whereas direct sums of tensor-product blocks lead to arithmetic combinations of multiplicative overlaps. These two constructions are generally inequivalent.

A particularly simple symmetric case is obtained by taking direct sums of tensor powers. If one considers blocks of the form \(\varrho_j^{\otimes \ell}\) and their direct sum, the induced overlap is proportional to
\begin{equation}
\frac{1}{k}\sum_{j=1}^k \left(\frac{\Sigma_U(\varrho_{j})}{D}\right)^\ell,
\end{equation}
so that the associated intensive logarithmic quantity is exactly the generalized-mean rugosity \(\mathfrak{R}_p\) with \(p=\ell\). In this sense, positive integer values of \(p\) admit a natural realization as direct sums of tensor powers. The general mixed-composition formulas are expanded in Appendix~\ref{app:mixed}.

The generalized-mean hierarchy helps organize the conceptual role of the different candidates. The negative-$p$ framework places increasing weight on the smallest overlaps. In the limit $p\to-\infty$, the minimum rugosity $\mathfrak{R}_{\min}$ becomes a fully worst-case quantity, since it asks how well the highest-texture state in the family can be aligned with the common textureless direction. At $p=0$, the geometric rugosity $\mathfrak{R}_G$ still penalizes small overlaps strongly, but does so in a smoother multiplicative way. At $p=1$, the arithmetic rugosity $\mathfrak{R}_A$ becomes softer, since it depends only on the average overlap. By contrast, for positive $p>1$, the hierarchy increasingly rewards large overlaps, and in the limit $p\to+\infty$ the quantity $\mathfrak{R}_{\max}$ depends only on the lowest-texture state in the family.

From this perspective, $\mathfrak{R}_G$, $\mathfrak{R}_A$, and $\mathfrak{R}_{\min}$ are not isolated constructions, but distinguished members of a broader generalized-mean hierarchy. The arithmetic rugosity is the soft framework, the geometric rugosity is an intermediate multiplicative framework, and the minimum rugosity is the most conservative one. The quantity $\mathfrak{R}_{\max}$, while mathematically natural as the optimistic endpoint of the same hierarchy, is typically too weak to serve as a sensitive set quantifier, since it is governed only by the single most favorable state in the family. In particular, if the family contains a pure state, then $\mathfrak{R}_{\max}(\mathcal S)=0$.

The tensor-product and direct-sum frameworks should therefore be understood as natural \emph{realizations} of particular constructions of this family-level hierarchy, not as the origin of the definitions themselves. The arithmetic rugosity is naturally associated with a direct-sum or mean-state encoding of the family, whereas the geometric rugosity is naturally associated with a tensor-product encoding. The minimum rugosity $\mathfrak{R}_{\min}$ does not appear to be naturally singled out by a standard linear or multiplicative state-composition rule; rather, it is best interpreted as the robust extremal framework of the hierarchy. Likewise, $\mathfrak{R}_{\max}$ is best interpreted as its optimistic extremal framework.

This viewpoint suggests that the choice of $p$ is not merely technical but reflects the notion of family compatibility that one wishes to emphasize. Negative values of $p$ privilege uniform good performance across all states in the family, $p=0$ balances the contributions multiplicatively, $p=1$ privileges the average performance, and large positive values of $p$ emphasize the best represented members.

% *}}}
\subsection{Applications and examples} % {{{*

\subsubsection{$N-1$ textureless states and one arbitrary outlier}

Consider the family $\cS_{N,\sigma}:=\{f_1,f_1,\dots,f_1,\sigma\}$, with $N-1$ copies of the textureless state and one arbitrary density operator $\sigma$. Then, the mean state reads
\begin{equation}
\bar\varrho=\frac{(N-1)f_1+\sigma}{N}.
\end{equation}
% and therefore
% \begin{equation}
% T_{\mathrm{set}}(\cS_{N,\sigma})
% =
% 1-\lambda^{\uparrow}\!\left(\frac{(N-1)f_1+\sigma}{N}\right).
% \label{eq:Tsetexample}
% \end{equation}
% If $\ket{f_1}$ is an eigenvector of $\sigma$ associated with $\lambda^{\uparrow}(\sigma)$, then the expression simplifies to
% \begin{equation}
% T_{\mathrm{set}}(\cS_{N,\sigma})
% =
% 1-\frac{N-1+\lambda^{\uparrow}(\sigma)}{N}
% =
% \frac{1-\lambda^{\uparrow}(\sigma)}{N}.
% \end{equation}
% In this favorable case the outlier is diluted by a factor $1/N$.

And for the arithmetic rugosity, one has
\begin{equation}
\mathfrak{R}_A(\cS_{N,\sigma})
=
-\ln\lambda^{\uparrow}\!\left(\frac{(N-1)f_1+\sigma}{N}\right).
\label{eq:RAexample}
\end{equation}
If $\ket{f_1}$ is a maximal-eigenvalue eigenvector of $\sigma$, then
\begin{equation}
\mathfrak{R}_A(\cS_{N,\sigma})
=
-\ln\!\left(\frac{N-1+\lambda^{\uparrow}(\sigma)}{N}\right).
\end{equation}
Hence $\mathfrak{R}_A$ is soft with respect to outliers, because the overlap is averaged \emph{before} the logarithm, a single non-zero texture state contributes in a diluted way.

To analyze the remaining two rugosities, we can find their upper and lower bounds and check when they coincide. Writing $|\psi\rangle=U^{\dagger}|f_1\rangle$, we have
\begin{equation}
\bra{f_1}Uf_1U^{\dagger}\ket{f_1}=|\braket{f_1}{\psi}|^2,
\qquad
\bra{f_1}U\sigma U^{\dagger}\ket{f_1}=\bra{\psi}\sigma\ket{\psi}.
\end{equation}
Therefore
\begin{equation}
\mathfrak{R}_{min}(\cS_{N,\sigma})
=
-\ln\!\left[
\max_{|\psi\rangle}
\min\!\left\{ |\braket{f_1}{\psi}|^2,\,\bra{\psi}\sigma\ket{\psi}\right\}
\right].
\label{eq:Rinfexample}
\end{equation}

Choosing $|\psi\rangle=|f_1\rangle$ gives the upper bound
\begin{equation}
\mathfrak{R}_{min}(\cS_{N,\sigma})\le -\ln \bra{f_1}\sigma\ket{f_1}.
\end{equation}
On the other hand, since
\begin{equation}
\min\!\left\{ |\braket{f_1}{\psi}|^2,\,\bra{\psi}\sigma\ket{\psi}\right\}
\le \bra{\psi}\sigma\ket{\psi}
\le \lambda^{\uparrow}(\sigma),
\end{equation}
we obtain the lower bound
\begin{equation}
\mathfrak{R}_{min}(\cS_{N,\sigma})\ge -\ln \lambda^{\uparrow}(\sigma).
\end{equation}
If $|f_1\rangle$ is a maximal-eigenvalue eigenvector of $\sigma$, both bounds coincide and
\begin{equation}
\mathfrak{R}_{min}(\cS_{N,\sigma})=-\ln \lambda^{\uparrow}(\sigma).
\end{equation}
Unlike $\mathfrak{R}_A$, the minimum rugosity does \emph{not} dilute the outlier by $1/N$, which makes this rugosity quantifier the most ``sensitive'' to texture in the set of states.

Using the same parametrization, the geometric rugosity is
\begin{equation}
\mathfrak{R}_G(\cS_{N,\sigma})
=
-\ln\!\left[
\max_{|\psi\rangle}
\left(
|\braket{f_1}{\psi}|^{2(N-1)}\,\bra{\psi}\sigma\ket{\psi}
\right)^{1/N}
\right].
\label{eq:RGexample1}
\end{equation}
Equivalently,
\begin{equation}
\mathfrak{R}_G(\cS_{N,\sigma})
=
\min_{|\psi\rangle}
\frac{1}{N}\left[
-(N-1)\ln |\braket{f_1}{\psi}|^2
-\ln\bra{\psi}\sigma\ket{\psi}
\right].
\label{eq:RGexample2}
\end{equation}

The same benchmark yields simple bounds. Since $|\braket{f_1}{\psi}|^2\le 1$ and $\bra{\psi}\sigma\ket{\psi}\le \lambda^{\uparrow}(\sigma)$,
\begin{equation}
\left(
|\braket{f_1}{\psi}|^{2(N-1)}\,\bra{\psi}\sigma\ket{\psi}
\right)^{1/N}
\le \lambda^{\uparrow}(\sigma)^{1/N},
\end{equation}
which implies
\begin{equation}
\mathfrak{R}_G(\cS_{N,\sigma})\ge -\frac{1}{N}\ln \lambda^{\uparrow}(\sigma).
\end{equation}
Choosing $|\psi\rangle=|f_1\rangle$ gives
\begin{equation}
\mathfrak{R}_G(\cS_{N,\sigma})\le -\frac{1}{N}\ln \bra{f_1}\sigma\ket{f_1}.
\end{equation}
If $|f_1\rangle$ is a maximal-eigenvalue eigenvector of $\sigma$, then the bounds coincide and
\begin{equation}
\mathfrak{R}_G(\cS_{N,\sigma})=-\frac{1}{N}\ln \lambda^{\uparrow}(\sigma).
\end{equation}
Hence, in this favorable case, the geometric rugosity dilutes the outlier by a factor $1/N$, but still retains a genuinely logarithmic dependence on the outlier's best overlap.

Finally, the weighted arithmetic--geometric mean inequality gives the ordering
\begin{equation}
\mathfrak{R}_{min}(\cS_{N,\sigma})\ge \mathfrak{R}_G(\cS_{N,\sigma})\ge \mathfrak{R}_A(\cS_{N,\sigma}),
\end{equation}
and therefore also for the benchmark family $\cS_{N,\sigma}$. Thus $\mathfrak{R}_A$ is the softest family rugosity, $\mathfrak{R}_{min}$ is the most conservative one, and $\mathfrak{R}_G$ is in between.

With this example, one can see how each rugosity quantifier reacts when there is only one outlier, but for the sake of completeness, one may also consider the case when all states in the set are the same, that is, $\cS_{N,\psi}:=\{\psi,\dots,\psi\}$. By choosing a common unitary $U$ such that $U \psi U^\dagger=f_1$, all overlaps are equal to $1$, and therefore $\mathfrak{R}_A(\cS_{N,\psi})=\mathfrak{R}_{min}(\cS_{N,\psi})=\mathfrak{R}_G(\cS_{N,\psi})=0$.

\subsubsection{Orthogonal basis family}

Consider the family $S_{\mathrm{orth}}^{(N)}:=\{e_1,e_2,\ldots,e_N\}$, with $e_a:=|e_a\rangle\langle e_a|$, where $\{|e_a\rangle\}_{a=1}^N$ is an orthonormal basis of an $N$-dimensional subspace. If the basis is complete in the whole Hilbert space, then $N=D$.

Writing again $|\psi\rangle=U^\dagger|f_1\rangle$, one has

\begin{equation}
\frac{\Sigma_U(\varrho_a)}{D}=\langle \psi|e_a|\psi\rangle=|\langle e_a|\psi\rangle|^2=:p_a.
\end{equation}

Since the vectors $|e_a\rangle$ are orthonormal, the numbers $p_a$ satisfy $p_a\ge 0$ and $\sum_{a=1}^N p_a\le 1$, with equality whenever $|\psi\rangle$ lies in the span of the basis family.

For this family, the set texture and all three rugosities admit exact evaluations. Since
\begin{equation}
\bar\varrho=\frac{1}{N}\sum_{a=1}^N e_a=\frac{P}{N},
\end{equation}
where $P$ is the projector onto the span of $\{|e_a\rangle\}_{a=1}^N$, the nonzero eigenvalues of $\bar\varrho$ are all equal to $1/N$. Therefore
\begin{equation}
%T_{\mathrm{set}}(S_{\mathrm{orth}}^{(N)})=1-\frac{1}{N},
%\qquad
\mathfrak{R}_A(S_{\mathrm{orth}}^{(N)})=\ln N.
\end{equation}
If the family is a complete orthonormal basis of the whole Hilbert space, then $P=\mathbb I_D$, so that $\bar\varrho=\mathbb I_D/N$.

For the minimum rugosity, one obtains
\begin{equation}
\mathfrak{R}_{min}(S_{\mathrm{orth}}^{(N)})
=
-\ln\!\left[
\max_{|\psi\rangle}\min_a p_a
\right].
\end{equation}
Because $\sum_{a=1}^N p_a\le 1$, one necessarily has $\min_a p_a\le \frac{1}{N}$. This bound is saturated by the equal-superposition state 

\begin{equation}
|\psi_{\mathrm{opt}}\rangle=\frac{1}{\sqrt N}\sum_{a=1}^N e^{i\theta_a}|e_a\rangle,
\end{equation}

for which $p_a=1/N$ for all $a$. Hence
\begin{equation}
\mathfrak{R}_{min}(S_{\mathrm{orth}}^{(N)})=\ln N.
\end{equation}

For the geometric rugosity,
\begin{equation}
\mathfrak{R}_G(S_{\mathrm{orth}}^{(N)})
=
-\ln\!\left[
\max_{|\psi\rangle}
\left(\prod_{a=1}^N p_a\right)^{1/N}
\right].
\end{equation}
Now the arithmetic--geometric mean inequality gives
\begin{equation}
\left(\prod_{a=1}^N p_a\right)^{1/N}
\le
\frac{1}{N}\sum_{a=1}^N p_a
\le
\frac{1}{N}.
\end{equation}
Equality is attained for the same equal-superposition state, so
$\mathfrak{R}_G(S_{\mathrm{orth}}^{(N)})=\ln N.$
Thus the orthogonal-basis family provides a clean example in which all three rugosity candidates agree. In the simplest case,
$S=\{\ketbra{0}{0},\ketbra{1}{1}\}$, one finds
%T_{\mathrm{set}}(S)=\frac{1}{2},
%\qquad
$\mathfrak{R}_A(S)=\mathfrak{R}_{min}(S)=\mathfrak{R}_G(S)=\ln 2$.

% *}}}
% *}}}
\section{Closing Remarks} % {{{* 8
In this work, we have presented advances in consolidating and expanding the resource theory of quantum-state texture (QST). We used the grand sum as the central quantity underlying the texture measures considered here. The grand sum is simple to evaluate both theoretically and experimentally \cite{QSTexperiment}, and depends on the coherences in the chosen basis through the sum of their real parts.

After a preliminary contextualization, we addressed the simplest scenario of a single qubit, providing a full characterization of free operations and a geometric description of their action on the Bloch ball. We presented alternative proofs that the grand sum completely determines free convertibility between pure qubit states but provides only a necessary condition for general states, consistent with the complete qubit conversion criterion established in Ref.~\cite{lin_2026_quantumstate}. An explicit pair of incomparable states illustrates the partial ordering under free convertibility and the limitations of grand-sum-based quantification. Extending our characterization of free operations to qudits and obtaining complete conversion criteria for higher-dimensional mixed states remain important directions for the resource theory of QST.

In addition, we examined explicit families of free operations in arbitrary dimensions, distinguishing texture-preserving channels, such as specific unitary operations and certain Weyl and flip channels, from texture-decreasing amplitude damping channels. By choosing the textureless state as the target of relaxation, the latter provides concrete examples of texture depletion, including complete erasure in the full-damping limit. 

By focusing on the bipartite case, we established quantitative relations between the texture of a composite system and that of its subsystems. We determined rigorous bounds relating the global and reduced grand sums, including tight bounds for the product of the reduced grand sums at fixed global grand sum, and exhibited physical states saturating both limits. These results constrain how much texture remains accessible after discarding a subsystem. We also explored basis-independent bounds on distillable texture and their connection to a related purity measure.

Finally, we developed the concept of set texture, introducing a basis-independent formulation for families of states through optimization over a shared reference basis. The geometric and arithmetic rugosities arise naturally from tensor-product and direct-sum descriptions, respectively, and belong to a broader generalized-mean hierarchy that also includes a worst-case criterion. The examples considered illustrate how these measures retain different information about a family, particularly in their sensitivity to outliers. A natural extension is to consider direct sums of Hilbert spaces with different dimensions, as occur in symmetry-resolved decompositions. This would require specifying suitable reference states, sector weights, and allowed basis changes, opening a route to understanding how texture can be quantified within individual symmetry sectors and across the full system.

The results obtained strengthen the mathematical foundations of QST and broaden its scope from individual states to composite systems and families of preparations. They provide a basis for further investigations of texture as a diagnostic tool in quantum circuits, dynamical criticality, and quantum phase transitions.
%
% *}}}
% acknowledgments and funding %{{{*
\ack{J.A.d.L., A.F., and F.P. thank the Paraty Quantum Information School and Workshop (2025) for providing the starting point for this collaboration.}

\funding{J.A.d.L. and M.G. acknowledge support from SECIHTI (graduate scholarship and project CBF-2025-I-1548) and from UNAM-PAPIIT (project IG101324). F. P. acknowledges financial support from the Brazilian agencies Coordena\c{c}\~ao de Aperfei\c{c}oamento de Pessoal de N\'{\i}vel Superior (CAPES), Conselho Nacional de Desenvolvimento Cient\'{\i}fico e Tecnol\'ogico (CNPq), Funda\c{c}\~ao de Amparo \`a Pesquisa do Estado de S\~ao Paulo (FAPESP - Grant 2021/06535-0), and Funda\c{c}\~ao de Amparo \`a Ci\^encia e Tecnologia do Estado de Pernambuco (FACEPE - Grant BPP-0037-1.05/24). P.C.A. acknowledges Conselho Nacional de Desenvolvimento Cient\'{\i}fico e Tecnol\'ogico (CNPq - 350189/2026-9) and Projeto QUANTA UFPE, Convênio FINEP nº 01.24.0504.00.}

\bibliographystyle{unsrturl}
\bibliography{refs1.bib}

\begin{thebibliography}{10}

\bibitem{gour}
Eric Chitambar and Gilad Gour.
\newblock Quantum {{Resource Theories}}.
\newblock {\em Rev. Mod. Phys.}, 91(2):025001, April 2019.
\newblock \href {https://arxiv.org/abs/1806.06107} {\path{arXiv:1806.06107}},
  \href {https://doi.org/10.1103/RevModPhys.91.025001}
  {\path{doi:10.1103/RevModPhys.91.025001}}.

\bibitem{l1}
T.~Baumgratz, M.~Cramer, and M.~B. Plenio.
\newblock Quantifying {{Coherence}}.
\newblock {\em Phys. Rev. Lett.}, 113(14):140401, September 2014.
\newblock \href {https://doi.org/10.1103/PhysRevLett.113.140401}
  {\path{doi:10.1103/PhysRevLett.113.140401}}.

\bibitem{imag}
Alexander Hickey and Gilad Gour.
\newblock Quantifying the imaginarity of quantum mechanics.
\newblock {\em J. Phys. A: Math. Theor.}, 51(41):414009, September 2018.
\newblock \href {https://doi.org/10.1088/1751-8121/aabe9c}
  {\path{doi:10.1088/1751-8121/aabe9c}}.

\bibitem{ringbauer}
Martin Ringbauer, Thomas~R. Bromley, Marco Cianciaruso, Ludovico Lami,
  W.~Y.~Sarah Lau, Gerardo Adesso, Andrew~G. White, Alessandro Fedrizzi, and
  Marco Piani.
\newblock Certification and {{Quantification}} of {{Multilevel Quantum
  Coherence}}.
\newblock {\em Phys. Rev. X}, 8(4):041007, October 2018.
\newblock \href {https://doi.org/10.1103/PhysRevX.8.041007}
  {\path{doi:10.1103/PhysRevX.8.041007}}.

\bibitem{qst}
Fernando Parisio.
\newblock Quantum-{{State Texture}} and {{Gate Identification}}.
\newblock {\em Phys. Rev. Lett.}, 133(26):260801, December 2024.
\newblock \href {https://doi.org/10.1103/PhysRevLett.133.260801}
  {\path{doi:10.1103/PhysRevLett.133.260801}}.

\bibitem{QSTexperiment}
Carlos H.~S. Vieira, Xinfang Nie, Dawei Lu, and Fernando Parisio.
\newblock Quantum-state texture dynamics: Theory and experiment, 2026.
\newblock URL: \url{https://arxiv.org/abs/2609.05248}, \href
  {https://arxiv.org/abs/2609.05248} {\path{arXiv:2609.05248}}.

\bibitem{salazar}
Roberto Salazar, Jakub Czartowski, Ricard~Ravell Rodr{\'i}guez, Grzegorz
  {Rajchel-Mieldzio{\'c}}, Pawe{\l} Horodecki, and Karol {\.Z}yczkowski.
\newblock Quantum {{Resource Theories}} beyond {{Convexity}}.
\newblock {\em Quantum}, 10:2104, May 2026.
\newblock \href {https://doi.org/10.22331/q-2026-05-13-2104}
  {\path{doi:10.22331/q-2026-05-13-2104}}.

\bibitem{generalization}
Alexander C.~B. Greenwood, Joseph~M. Lukens, Li~Qian, and Brian~T. Kirby.
\newblock Re-examining the role of state texture in gate identification and
  fixed-point resource theories.
\newblock {\em Phys. Rev. A}, July 2026.
\newblock \href {https://doi.org/10.1103/83pf-837s}
  {\path{doi:10.1103/83pf-837s}}.

\bibitem{tinggui}
Yiding Wang, Hui Liu, and Tinggui Zhang.
\newblock Quantifying quantum-state texture.
\newblock {\em Phys. Rev. A}, 111(4):042427, April 2025.
\newblock \href {https://doi.org/10.1103/PhysRevA.111.042427}
  {\path{doi:10.1103/PhysRevA.111.042427}}.

\bibitem{cao}
Chengyang Zhang, Zhihua Guo, Bingke Zheng, and Huaixin Cao.
\newblock Quantum-state texture measures via weight and {{Tsallis}} relative
  entropy.
\newblock {\em Physics Letters A}, 563:131056, December 2025.
\newblock \href {https://doi.org/10.1016/j.physleta.2025.131056}
  {\path{doi:10.1016/j.physleta.2025.131056}}.

\bibitem{kim}
Yuntao Cui, Zhaobing Fan, and Sunho Kim.
\newblock Quanutm-{{State Texture}} as a {{Resource}}: {{Measures}} and
  {{Nonclassical Interdependencies}}, October 2025.
\newblock \href {https://arxiv.org/abs/2508.07481} {\path{arXiv:2508.07481}},
  \href {https://doi.org/10.48550/arXiv.2508.07481}
  {\path{doi:10.48550/arXiv.2508.07481}}.

\bibitem{Muthu}
R.~Muthuganesan.
\newblock Quantum state texture: {{Geometric}} and theoretic information
  perspective.
\newblock {\em Physics Letters A}, 570:131263, February 2026.
\newblock \href {https://doi.org/10.1016/j.physleta.2025.131263}
  {\path{doi:10.1016/j.physleta.2025.131263}}.

\bibitem{yu}
Wei-Ming Cao, Zhi-Xiang Jin, Wei Chen, Yan-Ling Wang, and Bing Yu.
\newblock Generalized measures of quantum-state texture via relative entropies
  and their applications.
\newblock {\em J. Phys. A: Math. Theor.}, 59(4):045303, January 2026.
\newblock \href {https://doi.org/10.1088/1751-8121/ae3aea}
  {\path{doi:10.1088/1751-8121/ae3aea}}.

\bibitem{lei}
Xiangyu Chen and Qiang Lei.
\newblock Quantifying and detecting quantum-state texture.
\newblock {\em J. Phys. A: Math. Theor.}, 59(27):275301, July 2026.
\newblock \href {https://doi.org/10.1088/1751-8121/ae7b8d}
  {\path{doi:10.1088/1751-8121/ae7b8d}}.

\bibitem{mondal}
Shampa Mondal, Soumajit Das, Preeti Parashar, and Tamal Guha.
\newblock Free encoding capacity: {{A}} universal unit for quantum resources,
  January 2026.
\newblock \href {https://arxiv.org/abs/2601.23116} {\path{arXiv:2601.23116}},
  \href {https://doi.org/10.48550/arXiv.2601.23116}
  {\path{doi:10.48550/arXiv.2601.23116}}.

\bibitem{aditi}
Ayan Patra, Tanoy~Kanti Konar, Pritam Halder, and Aditi Sen(De).
\newblock Role of quantum state texture in probing resource theories and
  quantum phase transitions.
\newblock {\em Phys. Rev. A}, 113(2):022411, February 2026.
\newblock \href {https://doi.org/10.1103/y2wz-ypln}
  {\path{doi:10.1103/y2wz-ypln}}.

\bibitem{lucas}
Lucas~C. C{\'e}leri, Krissia Zawadzki, Ivan Medina, and Diogo~O.
  {Soares-Pinto}.
\newblock Quantum state texture of dynamical criticality, May 2026.
\newblock \href {https://arxiv.org/abs/2605.04161} {\path{arXiv:2605.04161}},
  \href {https://doi.org/10.48550/arXiv.2605.04161}
  {\path{doi:10.48550/arXiv.2605.04161}}.

\bibitem{bateries}
Yiding Wang, Xiaofen Huang, and Tinggui Zhang.
\newblock Correlations {{Between Quantum Battery Capacity}} and {{Quantum
  Resources}} for {{Two-qubit System}}.
\newblock {\em Advanced Quantum Technologies}, 9(4):e70289, 2026.
\newblock \href {https://doi.org/10.1002/qute.70289}
  {\path{doi:10.1002/qute.70289}}.

\bibitem{relativity}
Zhiming Huang, Lianghui Zhao, Yiyong Ye, Jinyi Wang, Zhenbang Rong, and Xiaokui
  Sheng.
\newblock Quantum-state texture for accelerated atoms interacting with a
  massive scalar field.
\newblock {\em Quantum Inf Process}, 24(12):388, December 2025.
\newblock \href {https://doi.org/10.1007/s11128-025-05010-2}
  {\path{doi:10.1007/s11128-025-05010-2}}.

\bibitem{lorentz}
Zhiming Huang.
\newblock Evolution of quantum state texture under lorentz transformations:
  Invariance, competition, and sensitivity in single-particle and bipartite
  entangled systems.
\newblock {\em Phys. Rev. A}, 113:062438, Jun 2026.
\newblock URL: \url{https://link.aps.org/doi/10.1103/hts6-k7d3}, \href
  {https://doi.org/10.1103/hts6-k7d3} {\path{doi:10.1103/hts6-k7d3}}.

\bibitem{BH}
Zhiming Huang.
\newblock Quantum-state texture in schwarzschild-de sitter spacetime.
\newblock {\em Physics Letters B}, 880:140770, 2026.
\newblock URL:
  \url{https://www.sciencedirect.com/science/article/pii/S0370269326006210},
  \href {https://doi.org/10.1016/j.physletb.2026.140770}
  {\path{doi:10.1016/j.physletb.2026.140770}}.

\bibitem{zhang2026}
Anqi Zhang, Yanze Zheng, Xiaofen Huang, and Tinggui Zhang.
\newblock Quantum correlations of tripartite mixed states in the black hole
  quantum atmosphere, 2026.
\newblock URL: \url{https://arxiv.org/abs/2607.14462}, \href
  {https://arxiv.org/abs/2607.14462} {\path{arXiv:2607.14462}}.

\bibitem{sup}
Xiaotong Wang, Shunlong Luo, and Yue Zhang.
\newblock Phase-{{Sensitive Superposition}} of {{Quantum States}}.
\newblock {\em Annalen der Physik}, 538(3):e70170, 2026.
\newblock \href {https://doi.org/10.1002/andp.70170}
  {\path{doi:10.1002/andp.70170}}.

\bibitem{ineq}
{ Richard Bellman} {Edwing Beckenbach}.
\newblock {\em Inequalities}.
\newblock Springer-Verlag, 1961.

\bibitem{Gottesman1997}
Daniel Gottesman.
\newblock {\em Stabilizer Codes and Quantum Error Correction}.
\newblock PhD thesis, California Institute of Technology, 1997.
\newblock \href {https://arxiv.org/abs/quant-ph/9705052}
  {\path{arXiv:quant-ph/9705052}}.

\bibitem{spinsphasetransition2026}
Heitor~P. Casagrande, Isaac~M. Carvalho, William~J. Munro, and Krissia
  Zawadzki.
\newblock Textures as a phase-transition probe for quantum spin chains.
\newblock 2026.
\newblock \href {https://arxiv.org/abs/2609.08120} {\path{arXiv:2609.08120}}.

\bibitem{cohen}
Claude {Cohen-Tannoudji}, Bernard Diu, and Franck Lalo{\"e}.
\newblock {\em Quantum {{Mechanics}}, {{Volume}} 1: {{Basic Concepts}},
  {{Tools}}, and {{Applications}}}.
\newblock John Wiley \& Sons, December 2019.

\bibitem{bethruskai_2002_analysis}
Mary Beth~Ruskai, Stanislaw Szarek, and Elisabeth Werner.
\newblock An analysis of completely-positive trace-preserving maps on
  {{M2}}{$<$}math{$><$}mtext{$>$}{{M}}{$<$}/mtext{$><$}msub{$><$}mi{$><$}/mi{$><$}mn{$>$}2{$<$}/mn{$><$}/msub{$><$}/math{$>$}.
\newblock {\em Linear Algebra and its Applications}, 347(1):159--187, May 2002.
\newblock \href {https://doi.org/10.1016/S0024-3795(01)00547-X}
  {\path{doi:10.1016/S0024-3795(01)00547-X}}.

\bibitem{bengtsson_2006_geometry}
Ingemar Bengtsson and Karol Zyczkowski.
\newblock {\em Geometry of {{Quantum States}}: {{An Introduction}} to {{Quantum
  Entanglement}}}.
\newblock Cambridge University Press, Cambridge, 2006.
\newblock \href {https://doi.org/10.1017/CBO9780511535048}
  {\path{doi:10.1017/CBO9780511535048}}.

\bibitem{watrous_2018_theory}
John Watrous.
\newblock {\em The {{Theory}} of {{Quantum Information}}}.
\newblock Cambridge University Press, Cambridge, 2018.
\newblock \href {https://doi.org/10.1017/9781316848142}
  {\path{doi:10.1017/9781316848142}}.

\bibitem{lin_2026_quantumstate}
Yufan Lin, Yu~Guo, Fei He, and Shuanping Du.
\newblock Quantum-state texture measure and texture transformation, September
  2026.
\newblock \href {https://arxiv.org/abs/2609.08163} {\path{arXiv:2609.08163}},
  \href {https://doi.org/10.48550/arXiv.2609.08163}
  {\path{doi:10.48550/arXiv.2609.08163}}.

\bibitem{bertlmann2008bloch}
Reinhold~A Bertlmann and Philipp Krammer.
\newblock Bloch vectors for qudits.
\newblock {\em J. Phys. A: Math. Theor.}, 41(23):235303, May 2008.
\newblock \href {https://doi.org/10.1088/1751-8113/41/23/235303}
  {\path{doi:10.1088/1751-8113/41/23/235303}}.

\bibitem{kibler_2008_variations}
Maurice~R Kibler.
\newblock Variations on a theme of {{Heisenberg}}, {{Pauli}} and {{Weyl}}*.
\newblock {\em J. Phys. A: Math. Theor.}, 41(37):375302, August 2008.
\newblock \href {https://doi.org/10.1088/1751-8113/41/37/375302}
  {\path{doi:10.1088/1751-8113/41/37/375302}}.

\bibitem{siewert2022orthogonal}
Jens Siewert.
\newblock On orthogonal bases in the {{Hilbert-Schmidt}} space of matrices.
\newblock {\em Journal of Physics Communications}, 6(5):055014, May 2022.
\newblock \href {https://arxiv.org/abs/2205.06035} {\path{arXiv:2205.06035}},
  \href {https://doi.org/10.1088/2399-6528/AC6F43}
  {\path{doi:10.1088/2399-6528/AC6F43}}.

\bibitem{Nielsen2010}
Michael~A. Nielsen and Isaac~L. Chuang.
\newblock {\em Quantum {{Computation}} and {{Quantum Information}}: 10th
  {{Anniversary Edition}}}.
\newblock Cambridge University Press, December 2010.

\bibitem{Basile2024}
Tom{\'a}s Basile, Jose~Alfredo {de Leon}, Alejandro Fonseca, Fran{\c c}ois
  Leyvraz, and Carlos Pineda.
\newblock Weyl channels for multipartite systems.
\newblock {\em Phys. Rev. A}, 109(3):032607, March 2024.
\newblock \href {https://doi.org/10.1103/PhysRevA.109.032607}
  {\path{doi:10.1103/PhysRevA.109.032607}}.

\bibitem{Fonseca2019}
Alejandro Fonseca.
\newblock High-dimensional quantum teleportation under noisy environments.
\newblock {\em Physical Review A}, 100:062311, 12 2019.
\newblock URL: \url{https://link.aps.org/doi/10.1103/PhysRevA.100.062311
  http://arxiv.org/abs/1908.01097}, \href
  {https://doi.org/10.1103/PhysRevA.100.062311}
  {\path{doi:10.1103/PhysRevA.100.062311}}.

\bibitem{Dutta16}
Arijit Dutta, Junghee Ryu, Wiesław Laskowski, and Marek Żukowski.
\newblock Entanglement criteria for noise resistance of two-qudit states.
\newblock {\em Physics Letters A}, 380(27):2191 -- 2199, 2016.
\newblock URL:
  \url{http://www.sciencedirect.com/science/article/pii/S0375960116301700},
  \href {https://doi.org/10.1016/j.physleta.2016.04.043}
  {\path{doi:10.1016/j.physleta.2016.04.043}}.

\bibitem{PhysRevA.102.012401}
Sumeet Khatri, Kunal Sharma, and Mark~M. Wilde.
\newblock Information-theoretic aspects of the generalized amplitude-damping
  channel.
\newblock {\em Phys. Rev. A}, 102:012401, Jul 2020.
\newblock URL: \url{https://link.aps.org/doi/10.1103/PhysRevA.102.012401},
  \href {https://doi.org/10.1103/PhysRevA.102.012401}
  {\path{doi:10.1103/PhysRevA.102.012401}}.

\bibitem{Chessa2023resonantmultilevel}
Stefano Chessa and Vittorio Giovannetti.
\newblock Resonant {M}ultilevel {A}mplitude {D}amping {C}hannels.
\newblock {\em {Quantum}}, 7:902, January 2023.
\newblock \href {https://doi.org/10.22331/q-2023-01-19-902}
  {\path{doi:10.22331/q-2023-01-19-902}}.

\bibitem{eigenvalues}
Y.-J. Han, Y.-S. Zhang, and G.-C. Guo.
\newblock Compatibility relations between the reduced and global density
  matrices.
\newblock {\em Phys. Rev. A}, 71(5):052306, May 2005.
\newblock \href {https://doi.org/10.1103/PhysRevA.71.052306}
  {\path{doi:10.1103/PhysRevA.71.052306}}.

\bibitem{Yu2016}
Chang-shui Yu, Si-ren Yang, and Bao-qing Guo.
\newblock Total quantum coherence and its applications.
\newblock {\em Quantum Information Processing}, 15(9):3773--3784, September
  2016.
\newblock \href {https://doi.org/10.1007/s11128-016-1376-y}
  {\path{doi:10.1007/s11128-016-1376-y}}.

\bibitem{Radhakrishnan2019}
Chandrashekar Radhakrishnan, Zhe Ding, Fazhan Shi, Jiangfeng Du, and Tim
  Byrnes.
\newblock Basis-independent quantum coherence and its distribution.
\newblock {\em Annals of Physics}, 409:167906, October 2019.
\newblock \href {https://doi.org/10.1016/j.aop.2019.04.020}
  {\path{doi:10.1016/j.aop.2019.04.020}}.

\bibitem{Jin2024}
Ming-Ming Du, Hong-Wei Li, Zhen Tao, Shu-Ting Shen, Xiao-Jing Yan, Xi-Yun Li,
  Wei Zhong, Yu-Bo Sheng, and Lan Zhou.
\newblock Basis-independent quantum coherence and its distribution under
  relativistic motion.
\newblock {\em Eur. Phys. J. C}, 84(8):838, August 2024.
\newblock \href {https://doi.org/10.1140/epjc/s10052-024-13164-z}
  {\path{doi:10.1140/epjc/s10052-024-13164-z}}.

\bibitem{Designolle2021}
S{\'e}bastien Designolle, Roope Uola, Kimmo Luoma, and Nicolas Brunner.
\newblock Set {{Coherence}}: {{Basis-Independent Quantification}} of {{Quantum
  Coherence}}.
\newblock {\em Phys. Rev. Lett.}, 126(22):220404, June 2021.
\newblock \href {https://doi.org/10.1103/PhysRevLett.126.220404}
  {\path{doi:10.1103/PhysRevLett.126.220404}}.

\end{thebibliography}

\appendix
\section{Generalized-mean limits and their relation to the family rugosities}\label{app:meanlimits} % {{{*

In this appendix we justify the generalized-mean limits underlying the arithmetic, geometric, and minimum constructions of family rugosity.

Let
\begin{equation}
M_p(x_1,\ldots,x_N):=
\left(
\frac{1}{N}\sum_{a=1}^N x_a^p
\right)^{1/p},
\qquad p\neq 0,
\label{eq:Mpdef_appendix}
\end{equation}
for positive real numbers \(x_1,\ldots,x_N>0\). The arithmetic mean corresponds to \(p=1\), the geometric mean arises in the limit \(p\to 0\), and the minimum arises in the limit \(p\to -\infty\).

\subsection{The limit \(p\to 0\): geometric mean}

\begin{proposition}
For \(x_1,\ldots,x_N>0\),
\begin{equation}
\lim_{p\to 0} M_p(x_1,\ldots,x_N)
=
\left(\prod_{a=1}^N x_a\right)^{1/N}.
\label{eq:MpGeom_appendix}
\end{equation}
\end{proposition}

\begin{proof}
Define $L(p):=\ln M_p(x_1,\ldots,x_N).$
Using Eq.~\eqref{eq:Mpdef_appendix}, this becomes
\begin{equation}
L(p)
=
\frac{1}{p}
\ln\!\left(
\frac{1}{N}\sum_{a=1}^N x_a^p
\right).
\label{eq:Lp_appendix}
\end{equation}
Now write $x_a^p=e^{p\ln x_a},$
so that
\begin{equation}
L(p)
=
\frac{1}{p}
\ln\!\left(
\frac{1}{N}\sum_{a=1}^N e^{p\ln x_a}
\right).
\end{equation}
As \(p\to 0\), each exponential tends to \(1\), and therefore
\begin{equation}
\frac{1}{N}\sum_{a=1}^N e^{p\ln x_a}\longrightarrow 1.
\end{equation}
Hence the numerator in Eq.~\eqref{eq:Lp_appendix} tends to \(\ln 1=0\), while the denominator tends to \(0\). We may therefore apply l'H\^opital's rule:
\begin{align}
\lim_{p\to 0} L(p)
&=
\lim_{p\to 0}
\frac{
\frac{d}{dp}
\ln\!\left(
\frac{1}{N}\sum_{a=1}^N e^{p\ln x_a}
\right)
}{
\frac{d}{dp}(p)
}
=
\lim_{p\to 0}
\frac{
\frac{1}{N}\sum_{a=1}^N (\ln x_a)e^{p\ln x_a}
}{
\frac{1}{N}\sum_{a=1}^N e^{p\ln x_a}
}.
\end{align}
Taking the limit \(p\to 0\) gives
$\lim_{p\to 0} L(p)
=
\frac{1}{N}\sum_{a=1}^N \ln x_a.$
Exponentiating both sides, we obtain
\begin{equation}
\lim_{p\to 0} M_p(x_1,\ldots,x_N)
=
\exp\!\left(
\frac{1}{N}\sum_{a=1}^N \ln x_a
\right)
=
\left(\prod_{a=1}^N x_a\right)^{1/N},
\end{equation}
which proves Eq.~\eqref{eq:MpGeom_appendix}.
\end{proof}

\subsection{The limit \(p\to -\infty\): minimum}

\begin{proposition}
For \(x_1,\ldots,x_N>0\),
\begin{equation}
\lim_{p\to -\infty} M_p(x_1,\ldots,x_N)
=
\min_{1\le a\le N} x_a.
\label{eq:MpMin_appendix}
\end{equation}
\end{proposition}

\begin{proof}
Let $m:=\min_{1\le a\le N} x_a.$
Since \(x_a>0\), we may write
$x_a=m\,r_a,\, r_a\ge 1,$
and at least one of the \(r_a\)'s is exactly equal to \(1\).

Substituting into Eq.~\eqref{eq:Mpdef_appendix}, we find
\begin{align}
M_p(x_1,\ldots,x_N)
&=
\left(
\frac{1}{N}\sum_{a=1}^N (mr_a)^p
\right)^{1/p}
=
\left(
m^p\frac{1}{N}\sum_{a=1}^N r_a^p
\right)^{1/p}
=
m
\left(
\frac{1}{N}\sum_{a=1}^N r_a^p
\right)^{1/p}.
\label{eq:MpFactor_appendix}
\end{align}
It therefore suffices to study
\begin{equation}
A_p:=
\left(
\frac{1}{N}\sum_{a=1}^N r_a^p
\right)^{1/p},
\qquad r_a\ge 1,
\end{equation}
with \(\min_a r_a=1\).

We first derive a lower bound. Since at least one \(r_a\) equals \(1\), one term in the sum is \(1^p=1\), and all other terms are nonnegative. Hence
\begin{equation}
\frac{1}{N}\sum_{a=1}^N r_a^p\ge \frac{1}{N},
\end{equation}
which implies
\begin{equation}
A_p\ge \left(\frac{1}{N}\right)^{1/p}.
\label{eq:ApLower_appendix}
\end{equation}
Now
\begin{equation}
\left(\frac{1}{N}\right)^{1/p}
=
\exp\!\left(\frac{\ln(1/N)}{p}\right)
\longrightarrow 1
\qquad (p\to -\infty).
\label{eq:LowerLimit_appendix}
\end{equation}

We next derive an upper bound. Since \(r_a\ge 1\) and \(p<0\), raising to a negative power reverses the inequality, so $r_a^p\le 1 \;\forall a$.

Therefore
\begin{equation}
\frac{1}{N}\sum_{a=1}^N r_a^p\le 1,
\end{equation}
and hence
\begin{equation}
A_p\le 1.
\label{eq:ApUpper_appendix}
\end{equation}

Combining Eqs.~\eqref{eq:ApLower_appendix}, \eqref{eq:LowerLimit_appendix}, and \eqref{eq:ApUpper_appendix}, we obtain
\begin{equation}
\left(\frac{1}{N}\right)^{1/p}\le A_p\le 1,
\end{equation}
and both bounds tend to \(1\) as \(p\to -\infty\). By the squeeze theorem, $\lim_{p\to -\infty} A_p=1$.
Returning to Eq.~\eqref{eq:MpFactor_appendix}, we conclude that
\begin{equation}
\lim_{p\to -\infty} M_p(x_1,\ldots,x_N)=m=\min_a x_a,
\end{equation}
which proves Eq.~\eqref{eq:MpMin_appendix}.
\end{proof}

\section{Mixed tensor-product and direct-sum encodings}\label{app:mixed}

In this appendix we examine composite encodings built from both tensor products and direct sums. These constructions clarify two points. First, mixed encodings do not in general produce a single generalized mean, but rather a nested aggregation of overlaps. Second, a particularly simple symmetric case provides a natural realization of the positive-integer framework of the generalized-mean hierarchy.

A preliminary remark is in order. If one literally repeats the same state \(\varrho\) in every slot of a mixed construction, then the resulting overlap is built from repeated copies of the same scalar
\begin{equation}
\Sigma_U(\varrho):= D \bra{f_1}U\varrho U^\dagger\ket{f_1},
\end{equation}
and the whole construction collapses to the optimized single-state quantity
\begin{equation}
-\ln\max_U \frac{\Sigma_U(\varrho)}{D}=-\ln \lambda^{\uparrow}(\varrho).
\end{equation}
Thus, in order to obtain a genuinely new family-level quantity, one must allow different states in different slots.

We therefore consider a two-index family \(\{\varrho_{rj}\}\), with
$r=1,\ldots,\ell,
\;
j=1,\ldots,k,$ 
and define
$\Sigma_U(\varrho_{rj}):= D \bra{f_1}U\varrho_{rj}U^\dagger\ket{f_1}.$

\subsection{Tensor products of direct-sum blocks}

Consider first the mixed encoding
$\Omega_{\oplus\to\otimes}
:=
\bigotimes_{r=1}^{\ell}
\left(
\varrho_{r1}\oplus\cdots\oplus\varrho_{rk}
\right),$ 
acting on the Hilbert space
$\bigotimes_{r=1}^{\ell}
\left(
\mathcal H\oplus\cdots\oplus\mathcal H
\right).$

If the same single-system basis change is applied to every block, the corresponding unitary is
\begin{equation}
U_{\oplus\to\otimes}
=
\bigotimes_{r=1}^{\ell}
(U\oplus\cdots\oplus U),
\end{equation}
and the natural reference vector is
\begin{equation}
\ket{F_{\oplus\to\otimes}}
=
\bigotimes_{r=1}^{\ell}
\frac{1}{\sqrt{k}}
(\ket{f_1}\oplus\cdots\oplus\ket{f_1}).
\end{equation}
A direct computation gives
\begin{equation}
\bra{F_{\oplus\to\otimes}}
U_{\oplus\to\otimes}
\Omega_{\oplus\to\otimes}
U_{\oplus\to\otimes}^\dagger
\ket{F_{\oplus\to\otimes}}
=
\prod_{r=1}^{\ell}
\left(
\frac{1}{k}\sum_{j=1}^k \frac{\Sigma_U(\varrho_{rj})}{D}
\right).
\label{eq:mixedGA}
\end{equation}
Thus this mixed construction yields a \emph{geometric combination of arithmetic averages}. Its natural intensive logarithmic version is
\begin{equation}
\mathfrak{R}_{\oplus\to\otimes}
:=
-\frac{1}{\ell}
\ln\!\left[
\max_U
\prod_{r=1}^{\ell}
\left(
\frac{1}{k}\sum_{j=1}^k \frac{\Sigma_U(\varrho_{rj})}{D}
\right)
\right].
\end{equation}
Equivalently,
\begin{equation}
\mathfrak{R}_{\oplus\to\otimes}
=
-\ln\!\left[
\max_U
\left(
\prod_{r=1}^{\ell}
A_r(U)
\right)^{1/\ell}
\right],
\qquad
A_r(U):=\frac{1}{k}\sum_{j=1}^k \frac{\Sigma_U(\varrho_{rj})}{D}.
\end{equation}

\subsection{Direct sums of tensor-product blocks}

Now consider the opposite order,
\begin{equation}
\Omega_{\otimes\to\oplus}
:=
\left(
\varrho_{11}\otimes\cdots\otimes\varrho_{1\ell}
\right)
\oplus
\cdots
\oplus
\left(
\varrho_{k1}\otimes\cdots\otimes\varrho_{k\ell}
\right).
\end{equation}
The corresponding unitary is
\begin{equation}
U_{\otimes\to\oplus}
=
(U^{\otimes \ell})\oplus\cdots\oplus(U^{\otimes \ell}),
\end{equation}
and the natural reference vector is
\begin{equation}
\ket{F_{\otimes\to\oplus}}
=
\frac{1}{\sqrt{k}}
\bigl(
\ket{f_1}^{\otimes \ell}\oplus\cdots\oplus\ket{f_1}^{\otimes \ell}
\bigr).
\end{equation}
The induced overlap is
\begin{equation}
\bra{F_{\otimes\to\oplus}}
U_{\otimes\to\oplus}
\Omega_{\otimes\to\oplus}
U_{\otimes\to\oplus}^\dagger
\ket{F_{\otimes\to\oplus}}
=
\frac{1}{k}
\sum_{j=1}^k
\prod_{r=1}^{\ell} \frac{\Sigma_U(\varrho_{jr})}{D}.
\label{eq:mixedAG}
\end{equation}
Thus this second mixed construction yields an \emph{arithmetic combination of multiplicative overlaps}. The corresponding intensive logarithmic quantity is
\begin{equation}
\mathfrak{R}_{\otimes\to\oplus}
:=
-\frac{1}{\ell}
\ln\!\left[
\max_U
\frac{1}{k}
\sum_{j=1}^k
\prod_{r=1}^{\ell} \frac{\Sigma_U(\varrho_{jr})}{D}
\right].
\end{equation}

In general, the two nested constructions are inequivalent:
\begin{equation}
\prod_{r=1}^{\ell}
\left(
\frac{1}{k}\sum_{j=1}^k \frac{\Sigma_U(\varrho_{rj})}{D}
\right)
\neq
\frac{1}{k}
\sum_{j=1}^k
\prod_{r=1}^{\ell} \frac{\Sigma_U(\varrho_{jr})}{D}.
\end{equation}
They therefore define different notions of family compatibility.

\subsection{A symmetric case and the positive-\(p\) framework}

A particularly simple special case is obtained by taking direct sums of tensor powers,
\begin{equation}
\Omega_{\ell}
:=
\varrho_1^{\otimes \ell}
\oplus
\cdots
\oplus
\varrho_k^{\otimes \ell}.
\end{equation}
This is a special case of the previous construction with
$\varrho_{jr}=\varrho_j
\;
\text{for all }r=1,\ldots,\ell.$
Then Eq.~\eqref{eq:mixedAG} reduces to
\begin{equation}
\bra{F_{\otimes\to\oplus}}
U_{\otimes\to\oplus}
\Omega_{\ell}
U_{\otimes\to\oplus}^\dagger
\ket{F_{\otimes\to\oplus}}
=
\frac{1}{k}\sum_{j=1}^k \left(\frac{\Sigma_U(\varrho_{j})}{D}\right)^\ell,
\qquad
\Sigma_U(\varrho_{j}):= D \bra{f_1}U\varrho_jU^\dagger\ket{f_1}.
\end{equation}
The corresponding intensive logarithmic quantity is therefore
\begin{align}
\mathfrak{R}_{\ell}(\mathcal S)
&=
-\frac{1}{\ell}
\ln\!\left[
\max_U
\frac{1}{k}\sum_{j=1}^k \left(\frac{\Sigma_U(\varrho_{j})}{D}\right)^\ell
\right]=
-\ln\!\left[
\max_U
\left(
\frac{1}{k}\sum_{j=1}^k \left(\frac{\Sigma_U(\varrho_{j})}{D}\right)^\ell
\right)^{1/\ell}
\right].
\end{align}
This is exactly the generalized-mean rugosity \(\mathfrak{R}_p\) evaluated at
$p=\ell.$
Thus, positive integer values of \(p\) admit a natural realization as direct sums of tensor powers.

By contrast, if one considers the opposite symmetric construction
$(\varrho_1\oplus\cdots\oplus\varrho_k)^{\otimes \ell},$
then Eq.~\eqref{eq:mixedGA} gives
\begin{equation}
\left(
\frac{1}{k}\sum_{j=1}^k \frac{\Sigma_U(\varrho_{j})}{D}
\right)^\ell,
\end{equation}
and the corresponding intensive logarithmic quantity reduces simply to
\begin{equation}
-\ln\!\left[
\max_U
\frac{1}{k}\sum_{j=1}^k \frac{\Sigma_U(\varrho_{j})}{D}
\right]
=
\mathfrak{R}_A(\mathcal S).
\end{equation}
So repeated tensoring of the same direct-sum block does not generate a new framework of the generalized-mean hierarchy; it merely reproduces the arithmetic one.

These constructions show that mixed direct-sum/tensor-product encodings do not, in general, lead to one more simple generalized mean. Rather, they naturally produce \emph{nested means}. Nevertheless, the symmetric case of direct sums of tensor powers gives a particularly suggestive interpretation of the positive-integer framework of the generalized-mean hierarchy. In this sense, the hierarchy \(\mathfrak{R}_p\) is not only a formal interpolation between arithmetic, geometric, minimum, and maximum constructions, but also admits partial realization through suitably structured composite encodings.

% *}}}
\section{Derivation of the \cj{} matrix for single-qubit QST free operations} % {{{*
\label{app:free-qubit-choi} 

Here we derive the general form of the \cj{} matrix used in Sec.~\ref{sec:free-qubit-operations}. Let us define $F_{jk}=\dyad{f_j}{f_k}$. The unnormalized \cj{} matrix of \Eref{eq:cj_free_ops} can be written as a $2\times2$ block matrix:
\begin{equation}
  J_\Lambda =
  \mqty(
    \Lambda(F_{11}) & \Lambda(F_{12}) \\
    \Lambda(F_{21}) & \Lambda(F_{22})
  ).
\end{equation}
Since $J_\Lambda$ needs to be a Hermitian matrix, and considering the fixed-point condition $\Lambda(f_1)=f_1$:
\begin{equation}\label{eq:app-free-choi-blocks}
  J_\Lambda =
  \mqty(
    F_{11} & \Lambda(F_{12}) \\
    \qty[\Lambda\qty(F_{12})]^\dagger & \Lambda(F_{22})
  ).
\end{equation}

We next impose trace preservation, $\Tr[\Lambda(F_{jk})]=\Tr(F_{jk})=\delta_{jk}$, with $\delta_{jk}$ the Kronecker delta. In terms of the \cj{} matrix this reads
\begin{equation}
    \Tr_2\qty(J_\Lambda)
    =
    \mathds{1}.
    \label{eq:app-free-TP}
\end{equation}
Taking the partial trace of Eq.~\eqref{eq:app-free-choi-blocks} and evaluating \Eref{eq:app-free-TP} yields
\begin{equation}\label{eq:app-free-block-traces}
  \Tr\qty[\Lambda\qty(F_{12})] = 0,
  \qquad
  \Tr\qty[\Lambda\qty(F_{22})] = 1.
\end{equation}

A general complex $2\times2$ matrix with zero trace can therefore be parametrized as
\begin{equation}\label{eq:Lambda12}
  \Lambda\qty(F_{12}) =
  \mqty(
    a & b \\
    c & -a
  ),
\end{equation}
where $a$, $b$, and $c$ are complex. Similarly, the most general Hermitian $2\times2$ matrix with unit trace can be written as
\begin{equation}\label{eq:Lambda22}
  \Lambda\qty(F_{22}) =
  \mqty(
    1-\eta & q \\
    q^* & \eta
  ),
\end{equation}
where $\eta$ is real and $q$ is complex. 

Substituting $F_{11}$, and Eqs.~\eqref{eq:Lambda12}  and \eqref{eq:Lambda22} into \Eref{eq:app-free-choi-blocks}, and writing the resulting matrix in the ordered basis $\{\ket{f_1\otimes f_1},\ket{f_1\otimes f_2}, \ket{f_2\otimes f_1},\ket{f_2\otimes f_2}\}$ yields the most general Hermitian \cj{} matrix compatible with trace preservation and the fixed-point condition $\Lambda(f_1)=f_1$:
\begin{equation}
  J_\Lambda =
  \mqty(
    1 & 0 & a & b\\
    0 & 0 & c & -a \\
    a^* & c^* & 1-\eta & q \\
    b^* & -a^* & q^* & \eta
  ),
\end{equation}
% *}}}

\end{document}